\documentclass[12pt]{article}
\usepackage[utf8]{inputenc}
\usepackage{amsmath}
\usepackage{graphicx}
\usepackage{natbib}
\usepackage{url} 

\newcommand{\blind}{0}

\usepackage{framed}
\usepackage{xcolor} 
\usepackage{float}
 \usepackage{lineno}
\usepackage{hyperref}

\usepackage{amssymb}
\usepackage{amsthm}
\usepackage{amsmath} 
\usepackage{bm}

\usepackage{fancyvrb}

\usepackage{multirow}

\newcommand{\azul}{\color{black}}

\definecolor{mypurple}{rgb}{0.75, 0.0, 1.0}

\newcommand{\pkg}[1]{{\fontseries{b}\selectfont #1}}
\let\proglang=\textsf
\let\code=\texttt
\newcommand{\inla}{\pkg{R-INLA}}
\newcommand{\ve}[1]{{\bm{#1}}} 
\newcommand{\vekey}[1]{{\bm{#1}}} 

\newtheorem{theorem}{Theorem}[section]

\newtheorem{proposition}[theorem]{Proposition}

\DefineVerbatimEnvironment{Sinput}{Verbatim}{fontshape=sl}
\DefineVerbatimEnvironment{Soutput}{Verbatim}{}
\DefineVerbatimEnvironment{Scode}{Verbatim}{fontshape=sl}
\newenvironment{Schunk}{}{}
\DefineVerbatimEnvironment{Code}{Verbatim}{}
\DefineVerbatimEnvironment{CodeInput}{Verbatim}{fontshape=sl}
\DefineVerbatimEnvironment{CodeOutput}{Verbatim}{}

\setkeys{Gin}{width=0.8\textwidth }

\begin{document}

\def\spacingset#1{\renewcommand{\baselinestretch}%
{#1}\small\normalsize} \spacingset{1}


\if0\blind
{
  \title{\bf The Integrated Nested Laplace Approximation for fitting {Dirichlet regression models}}
  \author{Joaqu\'in Mart\'inez-Minaya \thanks{
    The authors gratefully acknowledge \textit{the Ministerio de Educaci\'on y Ciencia (Spain) for financial support (jointly financed by the European Regional Development Fund) via Research Grant PID2020-117979RB-I00 (DC and AL-Q), the Basque Government through the BERC 2018-2021 program and by the Ministry of Science, Innovation and Universities: BCAM Severo Ochoa accreditation SEV-2017-0718 and PID2020-115882RB-I00 research project (JM-M), the Canadian Natural Sciences and Engineering Research Council and the Canadian Research Chairs program (DS).}}\hspace{.2cm}\\
    Department of Applied Statistics and Operational Research, and Quality, \\ Universitat Polit\`ecnica de Val\`encia\\
    Finn Lindgren \\
    School of Mathematics, University of Edinburgh\\
    Antonio L\'opez-Qu\'ilez \\
    Department of Statistics and Operations Research, University of Valencia \\
    Daniel Simpson \\
     Department of Econometrics and Business Statistics, Monash University \\
    and \\
    David Conesa \\
    Department of Statistics and Operations Research, University of Valencia}
  \maketitle
} \fi

\if1\blind
{
  \bigskip
  \bigskip
  \bigskip
  \begin{center}
    {\LARGE\bf Title}
\end{center}
  \medskip
} \fi

\bigskip


\begin{abstract}
This paper introduces a Laplace approximation to Bayesian inference in Dirichlet regression models, which can be used to analyze a set of variables on a simplex exhibiting skewness and heteroscedasticity, without having to transform the data. These data, which mainly consist of proportions or percentages of disjoint categories, are widely known as compositional data and are common in areas such as ecology, geology, and psychology. We provide both the theoretical foundations and a description of how Laplace approximation can be implemented in the case of Dirichlet regression. The paper also introduces the package \pkg{dirinla} in the \pkg{R}-language that extends the \inla{} package, which can not deal directly with Dirichlet likelihoods. Simulation studies are presented to validate the good behaviour of the proposed method, while a real data case-study is used to show how this approach can be applied.
\end{abstract}

\noindent%
\vfill

\newpage
\spacingset{1.5} 

\section{Introduction}
\label{sec:intro}
\label{sec1}
{The use of regression models with multivariate or correlated responses has enormously increased in the last few years. Different R-packages have been developed to deal with them, \pkg{MCMCglmm} \citep{hadfield2010}, \pkg{SabreR} \citep{nowosad2018} or \pkg{mcglm} \citep{wagner2018}, also packages which uses copulas \citep{masarotto2017}. The use of Multinomial likelihood in the case of multivariate discrete response is one of the most popular \citep{monyai2016, odeyemi2019, piccini2019}.

Responses can also be continuous, such as Gaussian \citep{anderson1958} or compositional data \citep{aitchison2005,hijazi2009}. Of particular interest are compositional data which consist of proportions or percentages of disjoint categories summing up to one, and play an important role in many fields such as ecology \citep{kobal2017, douma2019}, geology \citep{buccianti2014, engle2014}, genomics \citep{tsilimigras2016,shi2016,washburne2017} or medicine \citep{dumuid2018, fairclough2018}.}

One of the biggest problems one has to face when dealing with models with multivariate or correlated responses is that of performing inference. Their own complexity makes statistical analysis complicated. In the case of compositional data, there are different approaches to deal with these additional complications. One method, due to \cite{aitchison1986}, is based in the idea that ``information in compositional vectors is concerned with relative, not absolute magnitudes'', and uses log-ratio analysis to deal with the unit-sum constraint \citep{aitchison1981,aitchison1982,aitchison1983,aitchison1984}. Dirichlet regression models \citep{hijazi2009} are another good way of analyzing compositional data. By using appropriate link functions, Dirichlet regression provides a GLM-like framework that {relates compositional data with other relevant variables of interest}. Beta regression can be considered a special, and effectively univariate, case of the former with only two categories.

Different packages have been implemented in \proglang{R} \citep{rmanual}  that analyze compositional data using Beta regression and Dirichlet regression, both under the frequentist \citep{cribari2010,maier2014} and the Bayesian paradigm. In the case of the latter, the largest challenge is implementing the posterior approximation. In particular, it has been implemented in \pkg{BayesX} \citep{klein2015}, \pkg{Stan} \citep{sennhenn2018}, \pkg{BUGS} \citep{VanderMerwe2018} and \pkg{R-JAGS} \citep{plummer2016}. These packages are mainly based on Markov chain Monte Carlo (MCMC) methods, which  construct a Markov chain whose stationary distribution converges to the posterior distribution. {\azul However, the computational cost of MCMC can be high}. On the other hand, the integrated nested Laplace approximation (INLA) methodology \citep{rue2009inla}, whose main idea is to approximate the posterior distribution using the Laplace integration method, has become an alternative to MCMC, guaranteeing a higher computational speed for Latent Gaussian models (LGMs).

The INLA methodology is now a well-established tool for Bayesian inference in several research fields, including ecology, epidemiology, econometrics and environmental science, {\azul and is implemented in the \inla{} package \citep{rue2017bayesian}}. Nevertheless, and spite of its availability for a large number of models, \inla{} does not allow to deal with compositional data when the number of categories is bigger than 2, the reason being that it is constructed for models with univariate responses.

Our objective in this work is twofold. {We present an expansion of the INLA method for the particular case of the Dirichlet regression, providing both its theoretical foundations and a description of how it can be implemented for its application, and we} introduce the package \pkg{dirinla} in the \proglang{R}-language that allows its practical use. To do so, the remainder of the paper is structured as follows.  Section \ref{sec:dirichlet} provides the basics of the Dirichlet regression, while Section \ref{sec:INLA} gives the necessary hints about LGMs and the INLA approach to follow the remainder of the paper. Section \ref{sec:inference} depicts the new approach, and Section \ref{sec:dirinla} introduces the \pkg{dirinla} package and how to use it. Simulation studies about the performance of the method introduced is presented in Section \ref{sec:simulation}, followed by an illustration of its use on real data in Section \ref{sec:real_example}. Finally, Section \ref{sec:conclusions} concludes.

\section[Dirichlet likelihood]{Dirichlet likelihood} \label{sec:dirichlet}
{In what follows we present both the Dirichlet distribution and the Dirichlet regression, while introducing some assumptions and the notation that will be used in the rest of the paper.}

\subsection{Dirichlet distribution}
The Dirichlet distribution is the generalization of the widely known Beta distribution, and it is defined by the following probability density,
\begin{linenomath}
\begin{equation}
	p(\ve{y} \mid \ve{\alpha})= \frac{1}{\text{B}(\ve{\alpha})} \prod_{c=1}^C y_c^{\alpha_c -1} \,,
    \label{dirichlet}
\end{equation}
\end{linenomath}
where $\ve{\alpha} = (\alpha_1, \ldots, \alpha_C)$ is known as the vector of shape parameters for each category, $\alpha_c>0$ $\forall c$, $y_c \in (0,1)$, $\sum_{c=1}^C y_c=1$, and $\text{B}(\ve{\alpha})$ is the Multinomial Beta function, which serves as the normalizing constant. The Multinomial Beta function is defined as $\prod_{c=1}^C \Gamma(\alpha_c)/ \Gamma(\sum_{c=1}^C \alpha_c)$. The sum of all $\alpha$'s, $\alpha_0=\sum_{c=1}^C \alpha_c$, is usually interpreted as a precision parameter ($\tau$). The Beta distribution is the particular case when $C=2$. In addition, each variable is marginally Beta distributed with $\alpha=\alpha_c$ and $\beta=\alpha_0-\alpha_c$.

Let $\ve{y} \sim \mathcal{D}(\ve{\alpha})$ denote a variable that is Dirichlet distributed. The expected values are $\text{E}(y_c)=\alpha_c/\alpha_0$, the variances are $\text{Var}(y_c)=[\alpha_c(\alpha_0 - \alpha_c)]/[\alpha_0^2(\alpha_0 + 1)]$ and the covariances are $\text{Cov}(y_c, y_{c'})=(-\alpha_c \alpha_{c'})/[\alpha_0^2(\alpha_0 + 1)]$.

\subsubsection{Dealing with zeros and ones} \label{sec:transformation}
{Dirichlet distributions, as it happens with Beta distributions, are defined in the open interval $(0,1)$. Nevertheless, in practice, data may come from the closed interval $[0, 1]$. Although converting compositional data into Multinomial data could be an option, there are many situations in real life (see references in the Introduction) where we have to use the continuous nature of the compositional data. In such cases, one option to deal with zeros and ones in Dirichlet distributions is to transform them \citep{maier2014}, in a similar way as it happens in Beta distributions \citep{smithson2006}. In particular:}
\begin{linenomath} \begin{equation} \label{eq:transformation}
	\ve{y}^* = \frac{\ve{y}(N - 1) + 1/C}{N} \,,
\end{equation} \end{linenomath}
{\azul with $N$ being the number of observations of the $C$ dimensional Dirichlet response}. This transformation compresses the data symmetrically around $0.5$, so extreme values are affected more than values lying close to $1/2$. Additionally, as it is pointed out in \cite{maier2014}, if $N \rightarrow \infty$ the compression vanishes, that is, larger data sets are less affected by this transformation. {\azul However, for small $N$, the transformation can introduce a noticeable bias. In the INLA software, a new method has been implemented for the Beta distribution that instead treats extreme observations as censored variables, but it is not clear how to extend this to the Dirichlet case.} 

From now on, we assume that {the variable takes} values in the open interval $(0,1)$, so that the transformation \eqref{eq:transformation} is not required.

\subsection{Dirichlet regression models}
Let $\vekey{Y}$ be a matrix with $C$ rows and $N$ columns denoting $N$ observations for the different categories $C$ of the $C$ dimensional response variable $\vekey{Y}_{\bullet n} \sim \mathcal{D}(\vekey{\alpha}_n)$. Let {\azul $\eta^*_{cn}$} be the linear predictor for the $n$th observation in the $c$th category, so {\azul $\vekey{\eta}^*$} is a matrix with $C$ rows and $N$ columns. Let $\vekey{V}^{(c)}$, $c=1, \ldots, C$, represents a matrix with dimension $N \times J_c$ that contains the covariate values for each individual and each category, so $\vekey{V}^{(c)}_{n \bullet}$ shows the covariate values for the $n$th observation and the $c$th category. Let $\vekey{\beta}$ be a matrix with $J_c$ rows and $C$ columns representing the regression coefficients in each dimension. Finally, let also {\azul $\omega_{cn}$} represents a realization of a random effect for the the $c$th category and the $n$th observation. Then, the relationship between the parameters of the Dirichlet distribution and the elements of the linear predictor (including random effects) is set up as:{\azul
\begin{linenomath}
\begin{equation}\label{eq:dirichlet_regression}
	g(\alpha_{cn}) = \eta^*_{cn} = \vekey{V}^{(c)}_{n\bullet} \vekey{\beta}_{\bullet c} + \omega_{cn}\,\,,
\end{equation}
\end{linenomath}
where $g(\cdot)$ is the link-function. As $\alpha_c>0$ for $c = 1,\ldots,C$, log-link  $g(\cdot) = \log(\cdot)$ is used. The regression coefficients $\vekey{\beta}_{\bullet c}$ are a column vector with $J_c$ elements.}

Dirichlet distributions can also be parametrized in terms of the mean $\mu_{cn}=\frac{\alpha_{cn}}{\sum_{n = 1}^N \alpha_{cn}}$ and the precision $\tau$. In this case, the relationship between the mean and the linear predictor is stated as:{\azul 
\begin{linenomath}
\begin{equation}\label{eq:dirichlet_regression_2}
\mu_{cn} = \frac{\exp{(\eta^{*\ve{\mu}}_{cn})}}{\sum_{c = 1}^C \exp{(\eta^{*\ve{\mu}}_{cn})}} = \frac{\exp{(\vekey{V}^{(c)}_{n\bullet} \vekey{\gamma}_{\bullet c} + \omega_{cn}^{\vekey{\mu}})}}{\sum_{c = 1}^C \exp{(\vekey{V}^{(c)}_{n\bullet} \vekey{\gamma}_{\bullet c} + \omega_{cn}^{\vekey{\mu}})}}  \,\,.
\end{equation}
\end{linenomath}}
The regression coefficients {\azul $\vekey{\gamma}_{\bullet c}$} are now a column vector with $J_c$ elements, while $\omega_{cn}^{\vekey{\mu}}$ represents now a realization of a random effect in this parametrization. Moreover, as one of the categories can be obtained as one less the sum of the rest, we can employ the same Multinomial logit strategy used in Multinomial regression: linear predictor of one category (base category) is set to zero, whereby this category is virtually omitted and becomes the reference.

The equivalence between the two parametrizations comes from the equalities{\azul
 \begin{linenomath}
\begin{equation}
\tilde{\ve{\gamma}}_{\bullet c} = \tilde{\ve{\beta}}_{\bullet c} - \tilde{\ve{\beta}}_{\bullet r}\,,
\end{equation}
\end{linenomath}
and
 \begin{linenomath}
\begin{equation}
\omega_{cn}^{\ve{\mu}} = \omega_{cn}- \omega_{rn}\,,
\end{equation}
\end{linenomath}
where $\tilde{\ve{\gamma}}_{\bullet c}$ is a column vector with $J \geq J_c$ elements that contains $J_c$ non-zero elements and $J - J_c$ zero-elements. $J$ dimension comes from a general matrix $\boldsymbol{V}$ with dimension $N \times J$ which contains all the covariates values ($J$) for all the observations used in all the categories. $\tilde{\ve{\beta}}_{\bullet c}$ is also a column vector with $J \geq J_c$ elements, and $\tilde{\ve{\beta}}_{\bullet r}$ and $\omega_{rn}$ are the regression coefficients and the realization of a random effect for the base category respectively.}

Note that we can include in the formula as many random effects as we consider and from a different nature: temporal, spatial, etc. {\azul But here, without loss of generality, we have added just one to show the equivalence between both parametrizations (proof in Appendix \ref{appendix_param})}. As the proposal presented in the following Sections relies on the fact that there are no other parameters to estimate in the observation likelihood, from now on we focus on the first parametrization presented.

Equation (\ref{eq:dirichlet_regression}) can be rewritten in a vectorized form. In particular, if {\azul
\[\vekey{\tilde{\eta}}=
  \underbrace{\begin{bmatrix}
    \vekey{\eta}^*_{\bullet 1} \\
    \vdots \\
    \vekey{\eta}^*_{\bullet N}
  \end{bmatrix}}_{CN \times 1} \, \
\]}  denotes a restructured linear predictor, {\azul being $\vekey{\eta}^*_{\bullet n}$ a column vector representing the linear predictor for the $n$th observation and all the categories,} the model in matrix notation is
\begin{linenomath}
\begin{equation}\label{eq:dirichlet_regression_matricial}
	\vekey{\tilde{\eta}} = \vekey{A} \vekey{x} \,,
\end{equation}
\end{linenomath}
where $\vekey{A}$ is the matrix properly constructed with the covariates values and vectors of 1s for the random effects, and $\vekey{x}$ a vector formed by the regression coefficients plus the realizations of the random effects. {\azul Posteriorly, we will refer to this vector as the Latent Gaussian random field.} 


\section[INLA for Latent Gaussian Models (LGMs)]{INLA for Latent Gaussian Models (LGMs)} \label{sec:INLA}

In this section, we start with a brief explanation about LGMs, the framework where we are going to fit the Bayesian Dirichlet regression (subsection \ref{subsec:LGM}), followed by the main idea of the Laplace approximation (subsection \ref{subsec:laplace}) and finishing with the INLA methodology (subsection \ref{subsec:INLA}).

\subsection{LGMs}\label{subsec:LGM}
The popularity of INLA stems from the fact that it allows for fast approximate inference for LGMs, which are a large class of models that include a lot of classically important models \citep{rue2005}. LGMs can be written as a three-stage hierarchical model in which observations $\ve{y}$ can be assumed to be conditionally independent given a latent Gaussian random field $\ve{x}$ and hyperparameters {\azul $\ve{\theta}_1$,}
\begin{linenomath}{\azul 
$$\ve{y} \mid \ve{x}, \ve{\theta}_1 \sim \prod_{n=1}^N p(y_n \mid x_n,\ve{ \theta}_1)\,.$$}
\end{linenomath}
The versatility of the model class relates to the specification of the latent Gaussian field
\begin{linenomath}{\azul
$$\ve{x} \mid \ve{\theta}_2 \sim \mathcal{N}(\ve{\mu}(\ve{\theta}_2), \ve{Q}^{-1}(\ve{\theta}_2))$$}
\end{linenomath}
which includes all the latent (non-observable) components of interest, such as fixed effects and random terms, describing the underlying process of the data. The hyperparameters {\azul $\ve{\theta}=(\ve{\theta}_1, \ve{\theta}_2)$} control the latent Gaussian field and/or the likelihood for the data.

The LGMs are a class generalising the large number of related variants of additive and generalized models. If the likelihood $p(y_n \mid x_n, \ve{\theta})$ such that ``$y_n$ only depends on its linear predictor $\eta_n$'' yields the generalized linear model setup, the set $\{x_n, n = 1, \ldots, N \}$ can be interpreted as $\eta_n$, being $\eta_n$ the linear predictor which is additive with respect to other effects,
\begin{linenomath}
\begin{equation} \label{eq:additive}
	\eta_n = \beta_0 + \sum_{j} v_{nj} \beta_j  + \sum_k \omega_{kn} \,,
\end{equation}
\end{linenomath}
where $\beta_0$ is the intercept, $\ve{v}$ represents the fixed covariates with linear effects $\{\beta_j\}$, and the terms $\{\ve{\omega}_{k}\}$ represent specific Gaussian processes. Each $\omega_{kn}$ is the contribution of the model components $\ve{\omega}_k$ to the $n$th linear predictor \citep{rue2017bayesian}. {If a Gaussian prior is assumed for the intercept and the parameters of the fixed effects}, the joint distribution of $\ve{x} = \{\ve{\eta}, \beta_0, \ve{\beta}, \ve{\omega}_1, \ve{\omega}_2, \ldots\}$ is \emph{a priori} Gaussian. This yields the latent field $\ve{x}$ in the hierarchical LGM formulation. The  hyperparameters $\ve{\theta}$ contain the non-Gaussian parameters of the likelihood and the model components. These parameters commonly include  variance, scale, or correlation parameters.

In many important cases, the latent field is not only Gaussian, but also sparse Gaussian Markov random field \citep[GMRF]{rue2005}. A GMRF is  a multivariate Gaussian random variable with additional conditional independence properties: $x_j$ and $x_j'$ are conditionally independent given the remaining elements if and only if the $(i,j)$ entry of the precision matrix is $0$. Implementations of the INLA method frequently use this property to speed up computation.

\subsection{Laplace Approximation}\label{subsec:laplace}
Laplace approximation \citep{barndorff1989} is a technique used to approximate integrals of the form
\begin{linenomath}
\begin{equation}
	I_n=\int \exp(nf(x))\,\mathrm{d}x .
\end{equation}
\end{linenomath}
The main idea is to approximate the target with a scaled Gaussian density that matches the value and the curvature of the target distribution at the mode and evaluate the integral using this Gaussian instead.
If $x_0$ is the point where $f(x)$ has its maximum, then
\begin{linenomath}
\begin{align}
	I_n & \approx  \int \exp(n(f(x_0) + \frac{1}{2} (x-x_0)^2 f''(x_0)))\,\mathrm{d}x \nonumber \\
	& =  \exp(nf(x_0)) \sqrt{\frac{2 \pi}{-nf''(x_0)}} = \tilde{I}_n \,.
\end{align}
\end{linenomath}
If $nf(x)$ is interpreted as the sum of log-likelihoods and $x$ as the unknown parameter, the Gaussian approximation will be very accurate as $n\rightarrow \infty$ under appropriate regularity conditions.

If we are interested in computing a marginal distribution $p(\gamma_1)$ from a joint distribution $p(\ve{\gamma})$, the Laplace approximation of the integral
$\int p(\ve{\gamma}) \,\mathrm{d}\ve{\gamma}_{-1}$ can be expressed as follows:
\begin{linenomath}
\begin{align}
	p(\gamma_1) & =\left.\frac{p(\ve{\gamma})}{p(\ve{\gamma}_{-1} \mid \gamma_1)}\right|_{\ve{\gamma}_{-1}=\ve{\gamma}^*_{-1}} \nonumber \\
	& \approx\left.\frac{p(\ve{\gamma})}{p_G(\ve{\gamma}_{-1}; \ve{\mu}(\gamma_1), \ve{Q}(\gamma_1))} \right|_{\ve{\gamma}_{-1}=\ve{\gamma}^*_{-1}= \ve{\mu}(\gamma_1)} \,,
\end{align}
\end{linenomath}
where the first equality holds for any valid $\ve{\gamma}^*_{-1}$, and the mean $\ve{\mu}(\gamma_1)$ and precision $\ve{Q}(\gamma_1)$ are the parameters of the multivariate Gaussian density derived from the derivatives of $\log p(\ve{\gamma})$ with respect to $\ve{\gamma}_{-1}$, for fixed $\gamma_1$. If the posterior is close to a Gaussian density, the results will be more accurate than if the posterior is very non-Gaussian. In this context, unimodality is necessary since the integrand is being approximated with a Gaussian at the mode $\ve{\gamma}^*_{-1}=\ve{\mu}(\gamma_1)$.

\subsection{INLA} \label{subsec:INLA}
The main idea of INLA approach is to approximate the posteriors of interest: the marginal posteriors for the latent field, $p(x_m \mid \ve{y})$, and the marginal posteriors for the hyperparameters, $p(\theta_k \mid \ve{y})$. These posteriors can be written as
\begin{linenomath}
\begin{align}
	p(x_m \mid \ve{y}) & =  \int p(x_m \mid \ve{\theta}, \ve{y}) p(\ve{\theta} \mid \ve{y}) \,\mathrm{d} \ve{\theta} \,, \label{eq:marginals1} \\
	p(\theta_k \mid \ve{y}) & =  \int p(\ve{\theta} \mid \ve{y}) \,\mathrm{d} \ve{\theta}_{-k} \label{eq:marginals2} \,.
\end{align}
\end{linenomath}
The nested formulation is used to compute $p(x_m \mid \ve{y})$ by approximating $p(x_m \mid \ve{\theta}, \ve{y})$ and $p(\ve{\theta} \mid \ve{y})$, and then using numerical integration to integrate out $\ve{\theta}$. Similarly, $p(\theta_k \mid \ve{y})$ can be computed by approximating $p(\ve{\theta} \mid \ve{y})$ and integrating out $\ve{\theta}_{-k}$.

The marginal posterior distributions in (\ref{eq:marginals1}) and (\ref{eq:marginals2}) are computed using the Laplace approximation presented in subsection \ref{subsec:laplace}. In \cite{rue2009inla} it is shown that the nested approach yields a very accurate approximation if applied to LGMs.

All this methodology can be used through \proglang{R} with the \inla{} package. For more details about \inla{} we refer the reader to {\azul \cite{blangiardo2015,zuur2017beginner,faraway2018,krainski2018,moraga2019,gomez-rubio2020}}, where practical examples and code guidelines are provided.

However, and despite the advantages of \inla{} implementation, there are some limitations {when we deal with multivariate responses. The Multinomial case is solved, as the Multinomial likelihood can be approximated using the Poisson trick. It consists of transforming the Multinomial likelihood into a Poisson likelihood with additional parameters \citep{baker1994}. However, there is no method available when we deal with compositional data}. In what follows, we propose an expansion of the INLA method {for a Dirichlet response variable.}

\section{{Inference in Dirichlet likelihoods}} \label{sec:inference}
The INLA methodology is a tool that allows us to deal with a wide range of LGMs. However, when a multivariate response is required and several linear predictors are needed to explain it, the implemented \inla{} methodology has some limitations. {In the particular case of Dirichlet likelihoods, the main idea to incorporate them in the \pkg{R-INLA} is first to approximate the effect of the log likelihood on the posterior using the Laplace approach and then convert the multivariate response data into observations that \pkg{R-INLA} can deal with. The remainder of the Section presents both the theoretical fundamentals to approximate the effect of the log-likelihood function $\log p(\ve{Y} \mid \ve{x}, \ve{\theta})$ in the posterior using the Laplace approximation that provides the conditioned independent Gaussian pseudo-observations, and then an algorithmic representation of the method.}

\subsection{Fundamentals of the approximation}

Let $\ve{\eta}_{n}:=\ve{\eta}^*_{\bullet n}$ denote the linear predictor corresponding to the $n$th observation $\ve{y}_n:=\ve{Y}_{ \bullet n }$. If $l(\ve{y} \mid \ve{x})$ represents $-\log p(\ve{y} \mid \ve{x})$ for any $\ve{y}$ and $\ve{x}$, then $l(\ve{y}_n \mid \ve{\eta}_n)=-\log p(\ve{y}_n\mid \ve{\eta}_{n})$ is the log-likelihood function expressed for the $n$th observation, being $\ve{y}_n$ and $\ve{\eta}_n$ vectors with $C$ components.
{\azul
Moreover, if $\ve{\eta}^0_{n}$ is a vector with dimension $C$, we express the gradient of $l(\ve{y}_n \mid \ve{\eta}_n)$ in $\ve{\eta}^0_{n}$ as ${\ve{g}^0_\ve{\eta}}_n=\nabla_{\ve{\eta}_n}(l)(\ve{\eta}^0_n, \ve{y}_n)$, and the Hessian of $l(\ve{y}_n \mid \ve{\eta}_n)$ in $\ve{\eta}^0_n$ as $\ve{H}^0_{\ve{\eta}_n}$. Depending on which is more computationally convenient, $\ve{H}^{0}_{\ve{\eta}_n}$ can be either the true Hessian ($\nabla^2_{\ve{\eta}_n}(l)(\ve{\eta}^0_n, \ve{y}_n)$) or the expected Hessian ($\text{E}_{\ve{y}_n \mid {\ve{\eta}_n}} (\nabla^2_{\ve{\eta}_n}(l)(\ve{\eta}^0_n, \ve{y}_n))$) in $\ve{\eta}^0_n$. Let $\ve{L}^{0}_n$ be the result of applying the Cholesky factorization to $\ve{H}^{0}_{\ve{\eta}_n}$, $\ve{H}^{0}_{\ve{\eta}_n} = \ve{L}^{0}_n (\ve{L}^{0}_n)^T$.}
\begin{theorem}{\label{theorem_laplace}}
 If the Laplace approximation method is applied to $l(\ve{y}_n \mid \ve{\eta}_n)$ with respect to {\azul $\ve{\eta}^0_n$}, then the vector
\begin{linenomath}{\azul 
\begin{align}\label{eq:z_var_definition}
    \ve{z}^0_n & := (\ve{L}^{0}_n)^T [\ve{\eta}^0_n - (\ve{H}^{0}_{{\eta}_n})^{-1} \ve{g}^{0}_{\ve{\eta}_n}] = (\ve{L}^{0}_n)^T \ve{\eta}^0_n - (\ve{L}^{0}_n)^{-1} \ve{g}^{0}_{\ve{\eta}_n} \,,
\end{align}}
\end{linenomath}
is conditionally independent Gaussian distributed
\begin{linenomath}{\azul
\begin{align} \label{eq:log_likelihood_taylor_1}
  l(\ve{y}_n \mid \ve{\eta}_n) & \approx  l(\ve{z}^0_n \mid \ve{\eta}_n) = Constant + \frac{1}{2}[\ve{z}^0_n-(\ve{L}^{0}_n)^T \ve{\eta}_n]^T  [\ve{z}^0_n-(\ve{L}^0_n)^T \ve{\eta}_n] \,,
\end{align}}
\end{linenomath}{\azul
{i.e., $\ve{z}^0_n \mid \ve{\eta}_n \sim \mathcal{N}((\ve{L}^{0}_n)^T \ve{\eta}_n, \ve{I}_{d})$ and $z^0_{cn} \mid \ve{\eta}_n \sim \mathcal{N}([(\ve{L}^{0}_n)^T \ve{\eta}_n]_c, 1)$, and
 the constant value of the expression is $l(\ve{y}_n \mid \ve{\eta}^0_n) - \frac{1}{2} (\ve{g}^0_{\ve{\eta}_n})^T (\ve{H}^{0}_{\ve{\eta}_n})^{-1} \ve{g}^{0}_{\ve{\eta}_n}$.}}

\end{theorem}
\begin{proof}
For proof of the theorem see Appendix \ref{append_likelihood_approx}.
\end{proof}

Theorem \ref{theorem_laplace} allows us to  convert the observation vector $\ve{y}_n$ {into Gaussian conditionally independent pseudo-observations {\azul$\ve{z}^0_n$.} More importantly, this theorem can be expanded to multiple observations. In particular, if we denote {\azul
\begin{linenomath}\[
  {\ve{\tilde{\eta}}^0}=
  \underbrace{\begin{bmatrix}
    \ve{\eta}^{*0}_{\bullet 1} \\
    \vdots \\
    \ve{\eta}^{*0}_{\bullet N}
  \end{bmatrix}}_{CN \times 1} \,, \
\ve{g}^{0}_{\ve{\tilde{\eta}}}=
  \underbrace{\begin{bmatrix}
    \ve{g}^0_1 \\
    \vdots \\
    \ve{g}^0_N
  \end{bmatrix}}_{CN \times 1} \,, \
\ve{L}^{0}=
  \underbrace{\begin{bmatrix}
    \ve{L}^{0}_1 & & 0   \\
    & \ddots & \\
    0 & & \ve{L}^{0}_N
  \end{bmatrix}}_{CN \times CN} \,, \
\ve{H}^{0}_\ve{\tilde{\eta}}=
  \underbrace{\begin{bmatrix}
    \ve{H}^{0}_{\ve{\eta}_1} & & 0   \\
    & \ddots & \\
    0 & & \ve{H}^{0}_{\ve{\eta}_N}
  \end{bmatrix}}_{CN \times CN} \,,
\]\end{linenomath}}
then the following proposition stands.

\begin{proposition}{\label{proposition_laplace}}
The matrix
\begin{linenomath}{\azul
\begin{align}\label{eq:z_var_definition_mult}
  \ve{\tilde{z}}^0 &:= (\ve{L}^{0})^T \ve{\tilde{\eta}}^0 - (\ve{L}^{0})^{-1} \ve{g}^{0}_{\ve{\tilde{\eta}}} \,
\end{align}}
\end{linenomath}
is conditionally independent Gaussian distributed by columns,
\begin{linenomath}{\azul
\begin{align}
  \ve{\tilde{z}}^0 \mid \ve{\tilde{\eta}} & \sim  \mathcal{N}((\ve{L}^{0})^T \ve{\tilde{\eta}}, \ve{I}_{CN})\,.
\end{align}}
\end{linenomath}
\end{proposition}
\begin{proof}
    For proof of this proposition see Appendix \ref{append_likelihood_approx}.
\end{proof}
This approximation has been constructed for a generic {\azul $\ve{\tilde{\eta}}^0$}, but, as we are interested in building a Gaussian approximation of the effect of the likelihood on the posterior distribution, {\azul $\ve{\tilde{\eta}}^0$} has been chosen as the posterior mode of $l(\ve{\tilde{\eta}} \mid \ve{Y})$. Then: {\azul 
\begin{equation}
p(\vekey{Y} \mid \vekey{x}, \vekey{\theta}) = p( \vekey{Y} \mid \ve{\tilde{\eta}}) \approx p(\ve{\tilde{z}}^0 \mid \ve{\tilde{\eta}}) = p(\ve{\tilde{z}}^0 \mid \vekey{x}, \vekey{\theta}) \,.
\end{equation}
The model posterior is factorized as:
\begin{equation}
p(\vekey{\theta}, \vekey{x} \mid \vekey{y}) = p(\vekey{\theta} \mid \vekey{y}) \cdot p(\vekey{x} \mid \vekey{y}, \vekey{\theta}) \,,
\end{equation}
and the approximation is factorized as:
\begin{equation}
\overline{p}(\vekey{\theta}, \vekey{x} \mid \vekey{\tilde{z}}^0) = \overline{p}(\vekey{\theta} \mid \ve{\tilde{z}}^0) \cdot \overline{p}(\vekey{x} \mid \ve{\tilde{z}}^0, \vekey{\theta}) \,.
\end{equation}}
Note that the approximation is constructed for a generic $\ve{\tilde{\eta}}$, in other words, this approximation is conditionally dependent on the linear predictor. The linear predictor can be formed by fixed effects or random effects. Here, in the implementation of the algorithm, we focus on the case where fixed effects and iid random effects are added to the linear predictor.

\subsection{The algorithm} \label{sec:inla_dirichlet}
In what follows, we depict the different steps to compute the marginal posterior distributions of the latent field, $p(\ve{x}_m \mid \ve{Y})$, and the marginal posterior distribution of the hyperparameters {\azul $p(\ve{\theta}_k \mid \ve{Y})$}. To obtain them, it is necessary to numerically find the mode of the posterior distribution of the linear predictor {\azul $\ve{\tilde{\eta}}^0$}. This can be done by means of an iterative method in a similar way to \pkg{inlabru} \citep{bachl2019}.
{\azul
We define a functional $f(\overline{p}_{\ve{u}})$ of the posterior distribution at $\ve{u}$ which generates a latent field configuration. Our aim is to find {\azul an invariant point} of the functional, so that $\ve{x}^0 = f(\overline{p}_{\ve{x}^0})$. The choice for $f(\cdot)$ is the joint conditional mode $f(\overline{p}_{\ve{x}}) =$ arg max$_{\ve{x}} \overline{p}_{\ve{x}} (\ve{x} \mid \ve{\tilde{z}}^0, \ve{\theta}^0)$, being $\ve{\theta}^0 =$ arg max$_{\ve{\theta}} \overline{p}_{\ve{\theta}} (\ve{\theta} \mid \ve{\tilde{z}}^0)$.
}
The final algorithm is:

\begin{enumerate}
\item Let $\ve{x}^1$ and $\ve{\theta}^1$ be initial candidate points for the latent variables and the hyperparameters.

\item \textbf{Locate a good candidate in the latent variables for the computation of the new pseudo-observations}. Compute the mode {\azul ($\ve{x}^2$) in $\ve{x}$ of $p(\ve{x} \mid \ve{Y}, \ve{\theta}^1)$ by means of a quasi-Newton method \citep{dennis1977} with line search strategy and Armijo conditions \citep{nocedal2006} in $-\log (p(\ve{x} \mid \ve{Y}, \ve{\theta}^1) \propto l(\ve{Y} \mid \ve{x}, \ve{\theta}^1) + l(\ve{x} \mid \ve{\theta}^1)$, being $p(\ve{x} \mid \ve{\theta}^1)$ multivariate Gaussian as we are in the LGMs context. $\ve{\tilde{\eta}}^1$} can be easily calculated from the expression $\ve{\tilde{\eta}} = \ve{A} \ve{x}$.

\item {\azul \textbf{Calculate the conditionally independent Gaussian pseudobservations $\ve{\tilde{z}}^1$.} In this case, at the modal configuration established $\ve{x}^1$, the Hessian matrix $\ve{H}^{1}_{\ve{\tilde{\eta}}}$ is computed. If the submatrix corresponding to the {$n$th individual} $\ve{H}^{1}_{\tilde{\ve{\eta}}_n}$ is not positive definite, the expected Hessian is used instead to guarantee a positive definite $\ve{H}^{1}_{\ve{\tilde{\eta}}}$. Following the approximation previously presented, the Cholesky factorization is computed in $\ve{H}^1_{\ve{\tilde{\eta}}}=  \ve{L}^{1} (\ve{L}^{1})^T$. {The gradient} ($\ve{g}^{1}_{\ve{\tilde{\eta}}}$) is also calculated in $\ve{\tilde{\eta}}^1$. According to the equation (\ref{eq:z_var_definition_mult}), the scale and rotation of the original observations are done to get the pseudo-observations $\ve{\tilde{z}}^1$.}

{\azul
\item \textbf{Call \pkg{R-INLA}}. As pseudo-observations $\ve{\tilde{z}}^1$ are conditionally independent Gaussian observations, we are able to call \pkg{R-INLA}. If in Step 2 algorithm has converged for a given tolerance, the posterior distributions obtained here are the one that we were looking for: $\overline{p}(\ve{\theta} \mid  \ve{\tilde{z}}^1)$ and $\overline{p}(\ve{x} \mid  \ve{\tilde{z}}^1)$. Else, repeat from step 1 defining $\ve{x}^3 = f(\overline{p}_{\ve{x}^2})$ as the new starting point for the latent variables, and $\ve{\theta}^2 =$ arg max$_{\ve{\theta}} \overline{p}_{\ve{\theta}} (\ve{\theta} \mid \ve{\tilde{z}}^1)$ the new starting point for the hyperparameters.}
\end{enumerate}

After depicting the complete method, we focus on an implementation in \proglang{R} \citep{rmanual} of this approximation.

\section{The R-package dirinla}{\label{sec:dirinla}}
In what follows we {present \pkg{dirinla}, an \proglang{R} \citep{rmanual} package developed to fit Dirichlet regression models. This package can be installed and upgraded via the repository {\azul \url{https://github.com/inlabru-org/dirinla}}. To show how it works, we present an example of a Dirichlet regression without random effect, although, they can be included in the model, as done in Section~\ref{sec:simulation}. Then, this Section is divided in three parts: the first one presents the necessary commands to perform a simulation from a Dirichlet regression model; the second one is devoted to show how to fit those models; and the last one depicts how to predict using the package} {\azul (An extended version of this example is available in the vignette of the R-package \pkg{dirinla}}). In particular, we firstly illustrate how to simulate {100} data points from a Dirichlet regression model with {three} different categories and one different covariate per category:
\begin{linenomath}
\begin{align}
  \ve{Y}_{\bullet n} & \sim  \text{Dirichlet}(\alpha_{1n}, \ldots, \alpha_{3n}) \,, n = 1, \ldots, 100, \nonumber \\
  \log(\alpha_{1n}) & =  \beta_{01} + \beta_{11} v_{1n}, \nonumber \\
  \log(\alpha_{2n}) & =  \beta_{02} + \beta_{12} v_{2n},  \\
  \log(\alpha_{3n}) & =  \beta_{03} + \beta_{13} v_{3n}, \nonumber
\end{align}
\end{linenomath}
{being the parameters that compose the latent field $\beta_{01}= -1.5$, $\beta_{02}=-2$, $\beta_{03}=0$ (the intercepts), and $\beta_{11}=1$, $\beta_{12}=2.3$, $\beta_{13}=-1.9$ (the slopes). Note that covariates are different for each category. This could be particularized for a situation where all of them are the same.}

For simplicity, covariates are simulated from a Uniform distribution on (0,1). 
To posteriorly fit the model, and following the structure of LGMs, Gaussian prior distributions are assigned with precision $10^{-4}$ to all the elements of the Gaussian field. 

\subsection{Data simulation}
This subsection presents an example of how simulation can be conducted using the functions of \pkg{dirinla}. 

First, we simulate the covariates from a \text{Uniform}(0,1):
\begin{Schunk}
{
\begin{Sinput}
R> N <- 100
R> V <- as.data.frame(matrix(runif((3) * N, 0, 1), ncol = 3))
R> names(V) <- paste0('v', 1:3)
\end{Sinput}
}
\end{Schunk}

We then define the formula that we want to fit to keep the values of the different categories in a list. This object will be used to construct the $\ve{A}$ matrix. We use the function \code{formula\_list()} from the package \pkg{dirinla}. 
\begin{Schunk}
\begin{Sinput}
R> formula <- y ~ 1 + v1 | 1 + v2 | 1 + v3
R> names_cat <- formula_list(formula)
\end{Sinput}
\end{Schunk}

The values for the parameters composing the latent field are assigned to conduct the simulation. As we have previously depicted, $\beta_{01}= -1.5$, $\beta_{02}=-2$, $\beta_{03}=0$ are the intercepts, and $\beta_{11}=1$, $\beta_{12}=2.3$, $\beta_{13}=-1.9$  are the slopes: {\azul
\begin{Schunk}
{
\begin{Sinput}
R> x <- c(-1.5, 1, -2, 2.3, 0, -1.9) 
\end{Sinput}
}
\end{Schunk}
}
We call the function \code{data\_stack\_dirich()} of the package \pkg{dirinla} to construct the $\ve{A}$ matrix presented in previous sections. This function uses the \code{inla.stack()} structure of the package \pkg{R-INLA}. As a consequence, the returning object is an \code{inla.stack} object. Observe that the arguments are the response variable \code{y} (in this case it has not been generated yet), the names of the categories \code{covariates}, a matrix with the values of the covariates \code{data}, the number of categories \code{d} and the {number of observations} \code{N}. The sparse matrix $\ve{A}$ is then computed. 
\begin{Schunk}
\begin{Sinput}
R> mus <- exp(x) / sum(exp(x))
R> C <- length(names_cat)
R> A_construct <-
+    data_stack_dirich(y = as.vector(rep(NA, N * C)),
+                      covariates = names_cat,
+                      data       = V,
+                      d          = C,
+                      n          = N)
\end{Sinput}
\end{Schunk}
The next step is to construct the linear predictor as $\ve{\tilde{\eta}} = \ve{A} \ve{x}$ using the parameters fixed in the latent field. Using the exponential transformation it is easy to get the parameters $\ve{\alpha}$ of the Dirichlet distribution:
\begin{Schunk}
\begin{Sinput}
R> eta <- A_construct %*% x
R> alpha <- exp(eta)
\end{Sinput}
\end{Schunk}

The last stage is to generate the response variable using the function \code{rdirichlet()} from \pkg{DirichletReg} \citep{maier2014}. {\azul The output is a matrix with the response variable summing their rows up to one.}
\begin{Schunk}
\begin{Sinput}
R> y <- rdirichlet(N, alpha)
\end{Sinput}
\end{Schunk}

\subsection{Fitting the model}
We now present the functions of \pkg{dirinla} needed to fit Dirichlet regression models, the main one being \code{dirinlareg}. This function carries out all the steps presented in Section \ref{sec:inference}, and its use (for the simulated data from previous subsection) is as follows:
\begin{Schunk}
{
\begin{Sinput}
R> model.inla <- dirinlareg(
+    formula  = y ~ 1 + v1 | 1 + v2 | 1 + v3 ,
+    y        = y,
+    data.cov = V,
+    prec     = 0.0001,
+    verbose  = FALSE)
\end{Sinput}
}
\end{Schunk}
\noindent where we have to specify the model \code{formula}, the response variable $\ve{Y}$ in a matrix format, the \code{data.frame} with the covariates \code{data.cov}, and the precision of the Gaussian prior (\code{prec}) for the latent field $\ve{x}$. If we want to follow the process step by step, we can add the instruction \code{verbose = TRUE}.

Once the model is fitted, we can summarize the posterior distribution of the fixed effects by means of the function \code{summary} applied to the object generated. This object belongs to the \code{dirinlaregmodel} class. Three model selection criteria are also displayed: Deviance Information Criterion \citep[DIC]{spiegelhalter2002}, Watanabe-Akaike information criteria \citep[WAIC]{gelman2014}, and the  mean of the logarithm of the conditional predictive ordinate \citep[LCPO]{gneiting2007}. Lastly, the number of observations and the number of categories are also depicted (Table \ref{table:example1}).
\begin{table}[H]{ {\azul
\begin{center} {
\caption{Output obtained when \code{summary} command is employed in a \code{dirinlaregmodel} object.}
\begin{tabular}{r | rrrrrrr}
  \hline
 & & mean & sd & 0.025quant & 0.5quant & 0.975quant & mode \\ 
  \hline
\multirow{ 2}{*}{y1} & intercept & -1.5542 & 0.2014 & -1.9497 & -1.5542 & -1.1591 & -1.5542 \\ 
  & v1 & 1.0181 & 0.3690 & 0.2936 & 1.0181 & 1.7420 & 1.0181 \\ 
   \hline
   \multirow{ 2}{*}{y2} & intercept & -1.7601 & 0.2329 & -2.2174 & -1.7601 & -1.3032 & -1.7601 \\ 
  & v2 & 1.9847 & 0.3971 & 1.2051 & 1.9847 & 2.7636 & 1.9847 \\ 
  \hline
  \multirow{ 2}{*}{y1} & intercept & -0.0593 & 0.2500 & -0.5501 & -0.0593 & 0.4312 & -0.0593 \\ 
  & v3 & -1.9470 & 0.4158 & -2.7633 & -1.9470 & -1.1314 & -1.9470 \\ 
 \hline
 \multicolumn{7}{l}{DIC = 1821.0247, WAIC = 1827.6268, LCPO = 913.8319}\\
 \multicolumn{7}{l}{Number of observations: 100} \\
 \multicolumn{7}{l}{Number of Categories: 3} \\
 \hline
\end{tabular}
\label{table:example1}}
\end{center}}}
\end{table}
{Using the implemented \code{plot} method, we obtain the marginal posterior distributions of the latent field for the different categories (Figure \ref{fig:posteriors_dirinla_1}). These marginals and a summary of each one are stored in \code{model.inla\$marginals\_fixed} and \code{model.inla\$summary\_fixed} (see \code{example.R} from the supplementary code to see the details of the code). The \code{plot} method also depicts the posterior predictive distribution in the simplex (Figure \ref{fig:prediction_tern_example}). }

\begin{figure}[H]
\begin{center}
	\includegraphics[width = \textwidth]{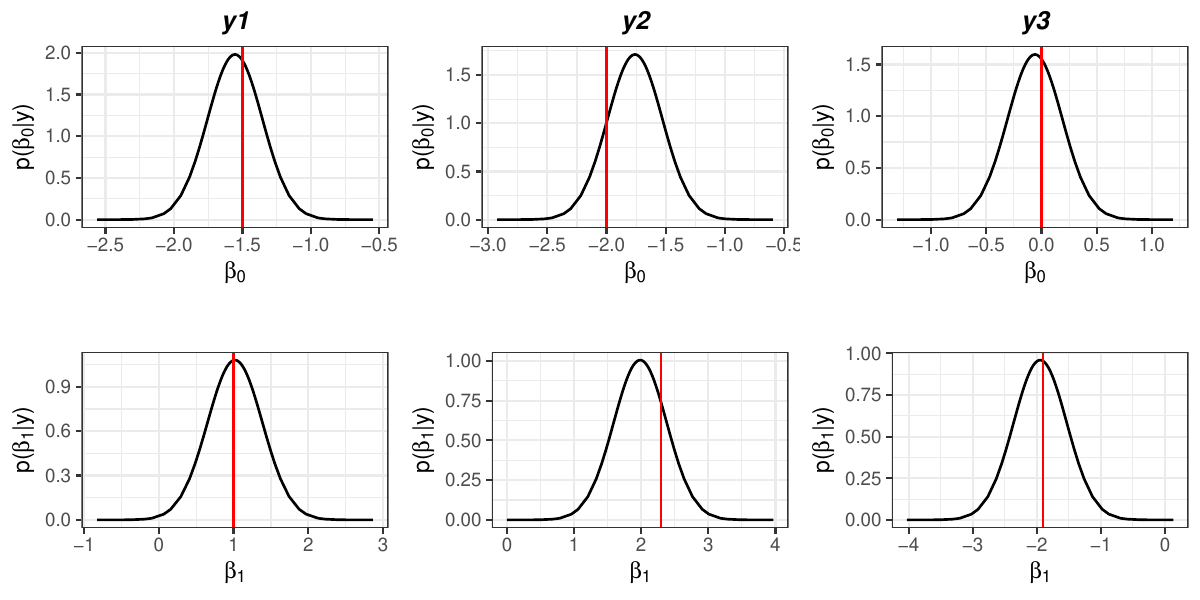}
\caption{Marginal posterior distributions of the latent field for the different categories. Real values are indicated with a red vertical line. The {sample size} is 50.}
\label{fig:posteriors_dirinla_1}
\end{center}
\end{figure}

\begin{figure}[H]
\begin{center}
	\includegraphics[width = 0.7\textwidth]{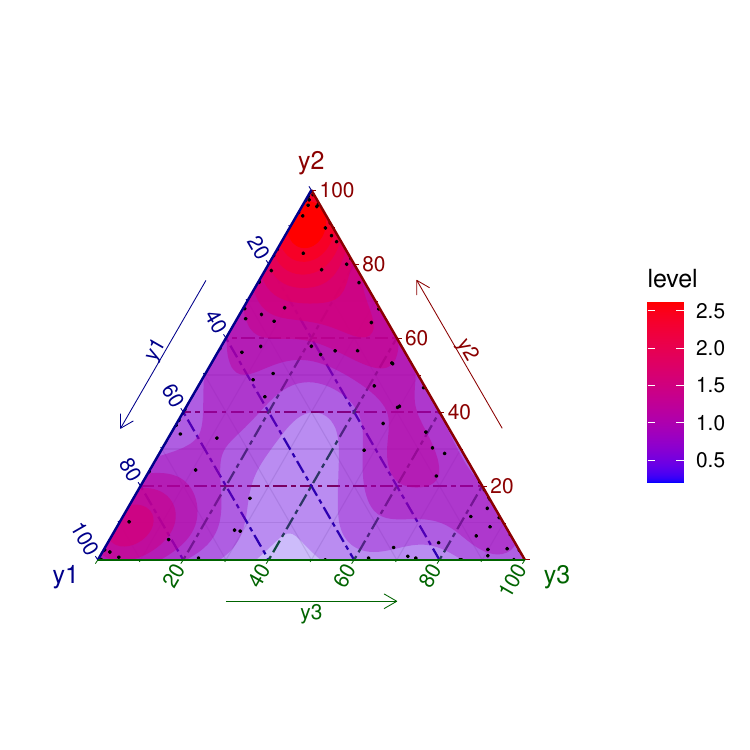}
\caption{{Posterior predictive distribution in the simplex. Points represent the original data.}}
\label{fig:prediction_tern_example}
\end{center}
\end{figure}

Finally, the posterior distribution for the scale parameters of the Dirichlet $\alpha$ can also be computed.
In particular, \code{model.inla\$marginals\_fixed} and \code{model.inla\$summary\_fixed} provide the marginals and a summary for each category. Mean parameters can also be obtained via \code{model.inla\$marginals\_means} or \code{model.inla\$summary\_means}, and similarly, \code{model.inla\$marginals\_precision} or \code{model.inla\$summary\_precision} provides the precision parameters.

\subsection{Prediction}
In most cases, practitioners want to be able to predict the composition of a new observation. The package also provides a function \code{predict} to compute posterior predictive distributions for new individuals. To show how it works, we now present how to predict for a value of {\texttt{v1 = 0.2}, \texttt{v2 = 0.5}, and \texttt{v3 = -0.1}}:
\begin{Schunk}
{
\begin{Sinput}
R> model.prediction <-
+  predict(model.inla,
+          data.pred.cov = data.frame(v1 = 0.2 ,
+                                     v2 = 0.5,
+                                     v3 = -0.1))
\end{Sinput}
}
\end{Schunk}
The resulting object also belongs to the \code{dirinlaregmodel} class. In a similar way as above, the elements \code{summary\_predictive\_alphas} and \code{marginals\_predictive\_alphas} describe the posterior predictive distribution for the scale parameters of the Dirichlet $\ve{\alpha}$, obtained for the new values of the covariates. {\azul In addition, means ($\mu$) and precisions ($\tau$) are available via \code{summary\_predictive\_means},  \code{marginals\_predictive\_means}, \code{summary\_predictive\_precision} and
\code{marginals\_predictive\_precision}}

\section{Simulation studies}\label{sec:simulation}
This section provides a comparison of the performance of the INLA approach for Dirichlet regression models using the \pkg{dirinla} package with the widely used method for Bayesian inference using MCMC algorithms, \pkg{R-JAGS} \citep{plummer2016}. The comparison was performed in four different simulated scenarios, and {\azul consists of comparisons of computational times and method accuracy.}

In all the cases to compare \pkg{dirinla} with \pkg{R-JAGS}, we employed three different methods to make inference: a standard application of the \pkg{R-JAGS} package with a number of iterations enough to {\azul achieve a given effective sample size}; the INLA methodology through the \pkg{dirinla} package; and a ``long'' application of the \pkg{R-JAGS} package (for simplicity called \pkg{long R-JAGS} from now on), in this case with a large amount of iterations to get really good representation of the posterior distributions. All computations were performed on a computer with a processor Intel(R) Core(TM) i5-10500 CPU @ 3.10GHz and 32 Gb RAM memory.

To compare the different ways of inference, we implemented four different strategies. The first one was to plot the approximate marginal posterior distributions of each parameter jointly with the real value. The second one consisted on comparing computational times between the three different methods. The remaining two strategies involved the computation of two different ratios:
\begin{eqnarray}\label{eq:ratio1_ratio2}
	ratio_1-\text{\pkg{method}} & = & (E(\phi_{\text{\pkg{method}}}) - E(\phi_{\text{\pkg{long R-JAGS}}}))/ SD(\phi_{\text{\pkg{long R-JAGS}}})  \,\,, \\
	ratio_2-\text{\pkg{method}} & = & SD(\phi_{\text{\pkg{method}}})/ SD(\phi_{\text{\pkg{long R-JAGS}}})  \,\,,
\end{eqnarray}
where \pkg{method} refers to the method we want to compare with \pkg{long R-JAGS}. It can be either \pkg{dirinla} or \pkg{R-JAGS}. $\phi$ is the parameter of interest.  $E(\cdot)$ represents the computed mean of the marginal posterior distribution, and $SD(\cdot)$ is the computed standard deviation of the posterior distribution. In what follows, we make a brief description of the four different scenarios with the aim of extending the explanation at a later stage.
\begin{enumerate}
	\item \textbf{Simulation 1}: the first scenario comprised the simulation of data coming from a Dirichlet regression with just intercepts in the model, with different data sizes.
	\item \textbf{Simulation 2}: the second scenario is similar to the previous one, but adding some complexity to the model. In particular, by adding a covariate per category, and checking with different data sizes.
	\item \textbf{Simulation 3}: in the third scenario, we show how the method behaves when random effects are added to the model. The simulations and model include random effects levels.
	\item \textbf{Simulation 4}: the objective of this last scenario is to check how our INLA approach for Dirichlet regression behaves when the number of categories increases. We simulated data from a Dirichlet regression with no slopes but increasing the number of categories.
\end{enumerate}

\subsection{Simulation 1} 
Our first setting is based on a Dirichlet regression with four categories and one parameter per category, the intercept, that is:
\begin{linenomath}
\begin{align} \label{eq:dirichlet_example1}
  \vekey{Y}_{\bullet n} & \sim  \text{Dirichlet}(\alpha_{1n}, \ldots, \alpha_{4n}) \,, n = 1, \ldots, N, \nonumber \\
  \log(\alpha_{1n}) & =  \beta_{01}, \nonumber \\
  \log(\alpha_{2n}) & =  \beta_{02}, \nonumber \\
  \log(\alpha_{3n}) & =  \beta_{03},  \\
  \log(\alpha_{4n}) & =  \beta_{04}. \nonumber
\end{align}
\end{linenomath}
Five different datasets of sizes $N= 50, 100$, $500, 1000, 10000$ with this structure were simulated letting $\beta_{0c}, c = 1,\ldots,4$ to be $-2.4$, $1.2$, $-3.1$ and $1.3$, respectively. We used vague prior distributions for the latent field {\azul($\boldsymbol{x} =$ \{$\beta_{0c}, c = 1,\ldots,4$\}). In particular, $p(x_m) \sim$ $\mathcal{N}(0, \tau = 0.0001)$, $m = 1, \ldots, 4$}. As the response values are not close to 0 and 1, no transformation was needed.

As above mentioned, for each simulated dataset, we employed three different methods to make inference: a standard application of the \pkg{R-JAGS} package with {$2000$ iterations, a {\azul burnin} of $200$, a thin number of $5$ and $3$ chains; the INLA methodology through the \pkg{dirinla} package; and \pkg{long R-JAGS}, using $1000000$ iterations with a {\azul burnin} of $100000$, a thin number of $5$ and $3$ chains. To compare methodologies, the strategies presented at the begining of the section were used.

As seen in Figure \ref{fig:posteriors1_100} (the rest of the plots are shown in the Simulation 1 of the supplementary material in \url{https://jmartinez-minaya.github.io/supplementary/CODA/tests.html}), the posterior densities have similar shape, with the new method tending to agree more with the \pkg{long R-JAGS} result than the more variable short run \pkg{R-JAGS} does, illustrating how our method reduces estimator variability, at the potential cost of a generally small bias.
\begin{figure}[H]{
\begin{center}
\begin{tabular}{c}
	\includegraphics[width = \textwidth]{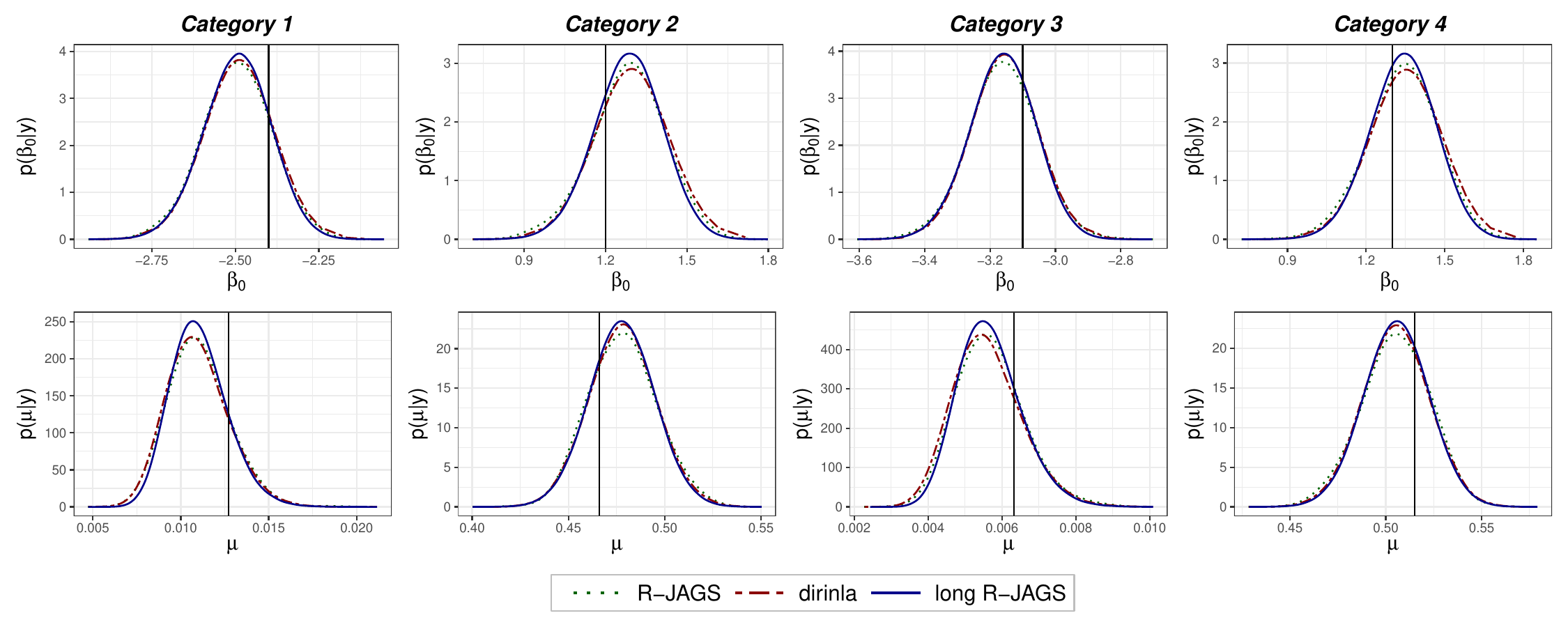} \\
\end{tabular}
\caption{Simulation 1: marginal posterior distributions of the latent field for the different categories, and using different methodologies: \pkg{R-JAGS}, \pkg{dirinla} and \pkg{long R-JAGS}, when the {sample size} is 100. {\azul Black vertical lines represent real values.}}
\label{fig:posteriors1_100}
\end{center}}
\end{figure}
{\azul With respect the computational effort needed to get those results, Figure \ref{fig:times_sim1} displays that, \pkg{dirinla} methodology has a faster computational speed for given data sizes,
for the given \pkg{R-JAGS} chain lengths.}
\begin{figure}[H]{
\begin{center}{\azul
	\includegraphics[width = 0.8\textwidth]{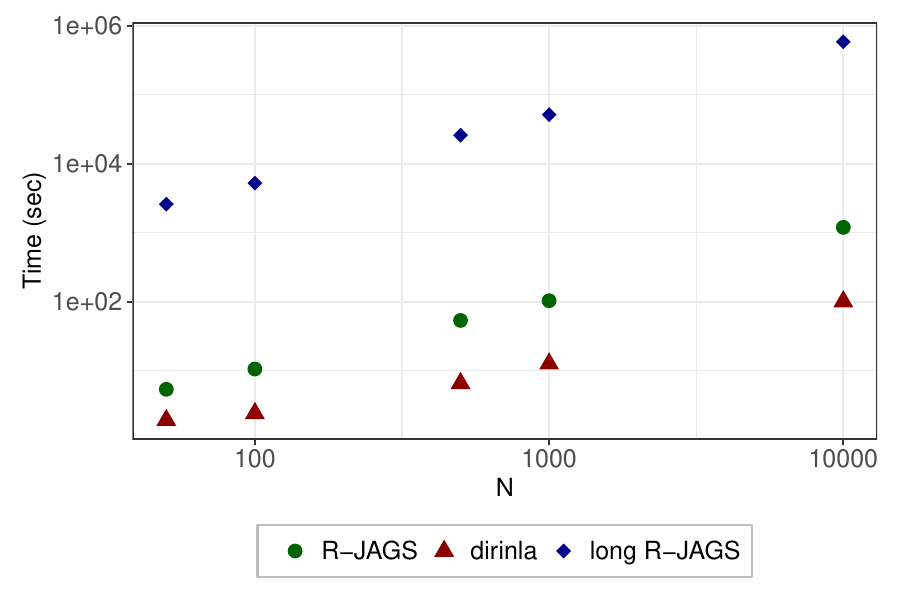} \\
\caption{Simulation 1: computational time in seconds for the different simulated data with N = 50, 100, 500, 1000 and 10000, and with the different methodologies: \pkg{R-JAGS}, \pkg{dirinla} and \pkg{long R-JAGS}.}
\label{fig:times_sim1}
}\end{center}}
\end{figure}
{\azul In Figure \ref{fig:ratios_sim1}}, we see that $ratio_1$-\pkg{dirinla} is so closed to 0 as happen with $ratio_1$-\pkg{R-JAGS}, and $ratio_2$-\pkg{dirinla} is always close to 1, similar to $ratio_2$-\pkg{R-JAGS}, meaning that using \pkg{dirinla} we obtain as good approximations as with \pkg{R-JAGS} reducing considerably the computational cost.
\begin{figure}[H]{
\begin{center}{\azul
	\includegraphics[width = \textwidth]{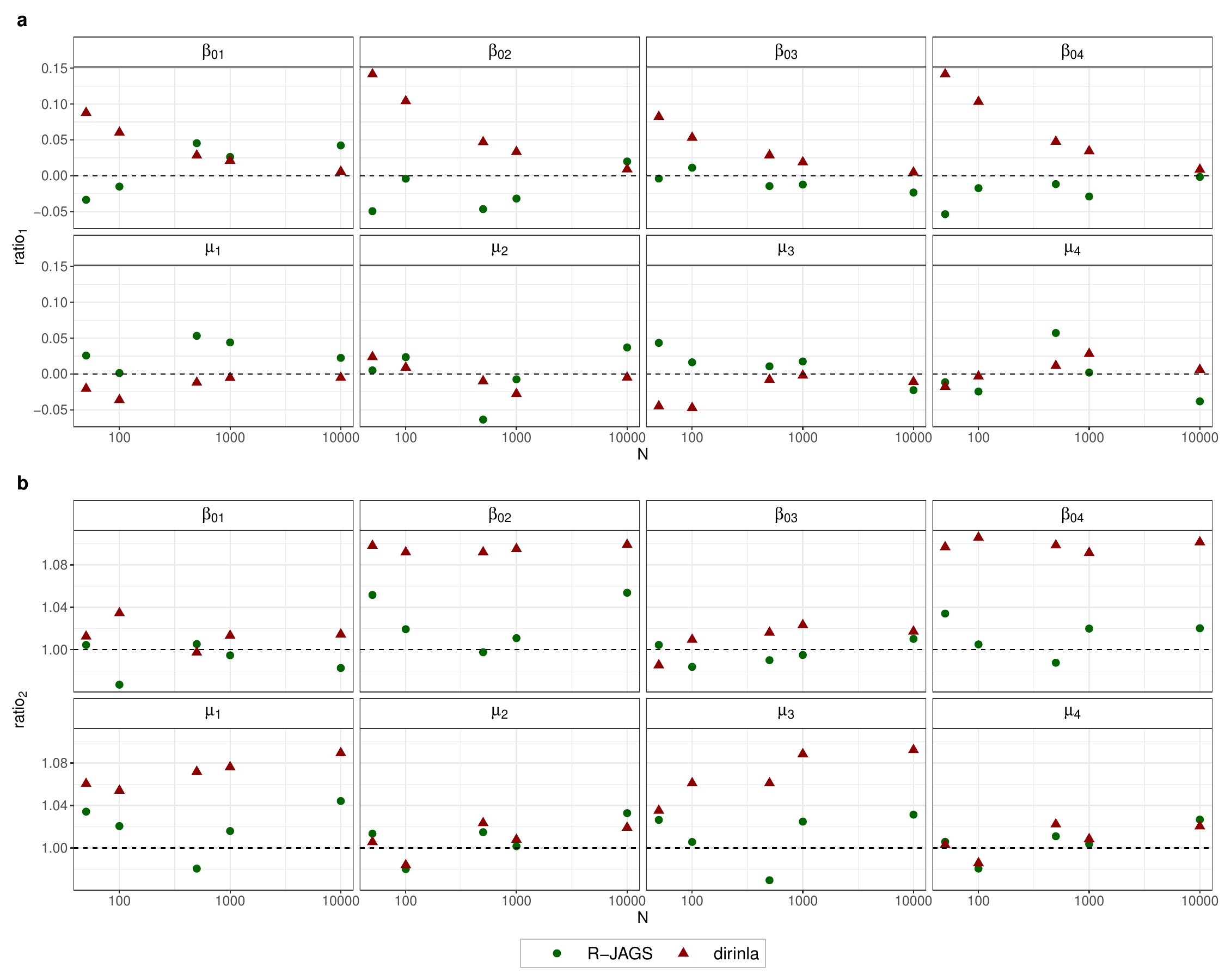} \\
\caption{Simulation 1: comparing accuracy between \pkg{dirinla} and \pkg{R-JAGS} computing two different measures: $ratio_1$ (a) and $ratio_2$ (b) for $\beta_{0c}$ and $\mu_{c}$, $c = 1, \ldots, 4$.}
\label{fig:ratios_sim1}
}\end{center}}
\end{figure}

\subsection{Simulation 2}
The second setting is based on a Dirichlet regression with a different covariate per category:
\begin{linenomath}
\begin{align}
  \vekey{Y}_{\bullet n} & \sim  \text{Dirichlet}(\alpha_{1n}, \ldots, \alpha_{4n}) \,, n = 1, \ldots, N, \nonumber \\
  \log(\alpha_{1n}) & =  \beta_{01} + \beta_{11} v_{1n}, \nonumber \\
  \log(\alpha_{2n}) & =  \beta_{02} + \beta_{12} v_{2n}, \nonumber \\
  \log(\alpha_{3n}) & =  \beta_{03} + \beta_{13} v_{3n},  \\
  \log(\alpha_{4n}) & =  \beta_{04} + \beta_{14} v_{4n}. \nonumber
\end{align}
\end{linenomath}
{Again, we simulated five different datasets of sizes $N= 50, 100$, $500, 1000, 10000$. We set values for $\beta_{0c}$ and $\beta_{1c}$, $c = 1, \ldots, 4$ to $-1.5, 1, -3, 1.5, 2, -3 , -1, 5$ respectively, and we simulated covariates from a Uniform distribution with mean in the interval $(0,1)$. We assigned vague prior distributions for the latent field {\azul($\boldsymbol{x}=$\{$\beta_{0c}, \beta_{1c}, c = 1,\ldots,4$\}). Particularly, $p(x_m) \sim$ $\mathcal{N}(0, \tau = 0.0001)$, $m = 1,\ldots, 8$}. As the data generated did not present zeros and ones, we did not use any transformation.}

As in the previous simulation, we employed the same three different inference methods with the same configurations and the four strategies used in the previous simulation. In Figure \ref{fig:posteriors_100} (see Simulation 2 in \url{https://jmartinez-minaya.github.io/supplementary/CODA/tests.html} for further plots), the posterior distributions are similar in both $\beta_{0c}$ and $\beta_{1c}$, $c = 1, \ldots, 4$.
\begin{figure}[H]
\begin{center}
\begin{tabular}{c}
	\includegraphics[width = \textwidth]{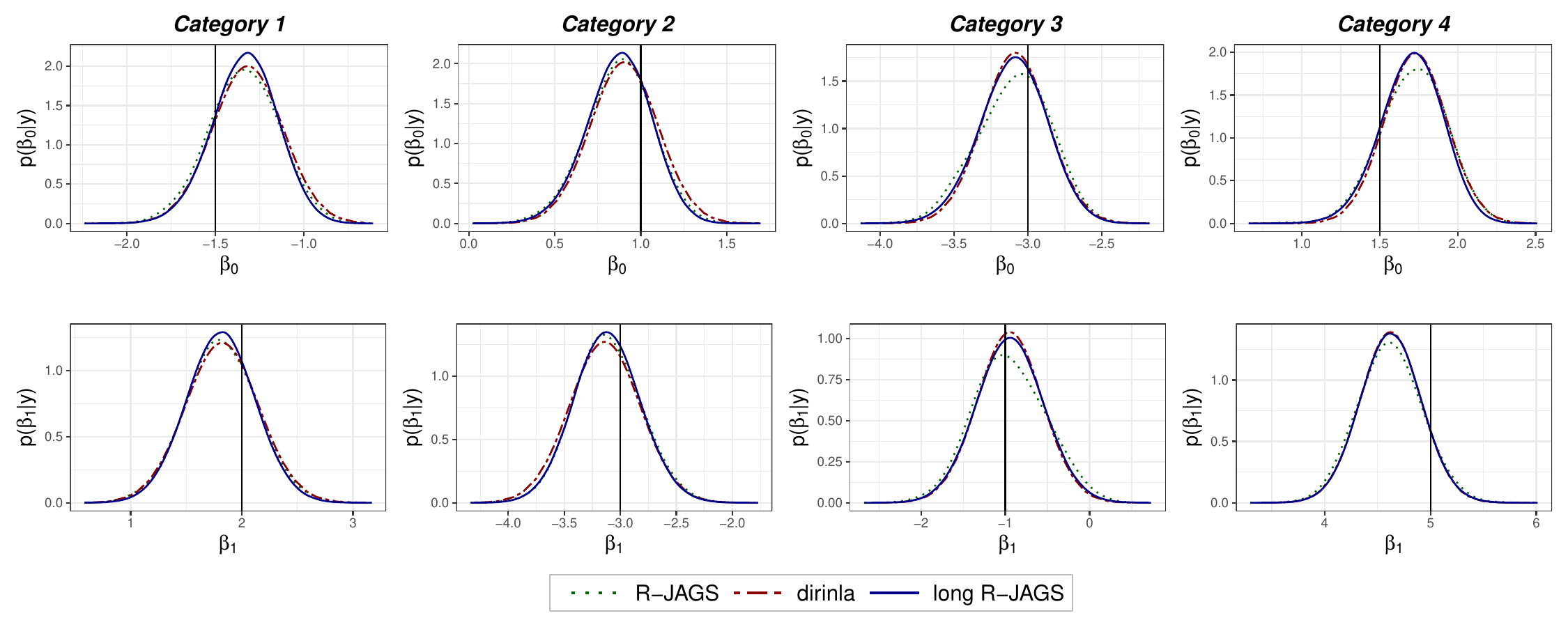} \\
\end{tabular}
\caption{Simulation 2: marginal posterior distributions of the latent field for the different categories, and using different methodologies \pkg{R-JAGS}, \pkg{dirinla} and \pkg{long R-JAGS}, when the {sample size} is 100. {\azul Black vertical lines represent real values.}}
\label{fig:posteriors_100}
\end{center}
\end{figure}
{\azul Figure \ref{fig:times_sim2} depicts that \pkg{dirinla} has a faster computational speed for given data sizes and for given \pkg{R-JAGS} chain lengths. In Figure \ref{fig:ratios_sim2}}, we show that $ratio_1$-\pkg{dirinla} and $ratio_1$-\pkg{R-JAGS} are so close to 0 in all cases, and $ratio_2$-\pkg{dirinla} and $ratio_2$-\pkg{R-JAGS} are always close to 1. Again, we conclude that using \pkg{dirinla} we get similar approximations to \pkg{R-JAGS}, while substantially reducing the computational cost.
\begin{figure}[H]{
\begin{center}{\azul
	\includegraphics[width = 0.8\textwidth]{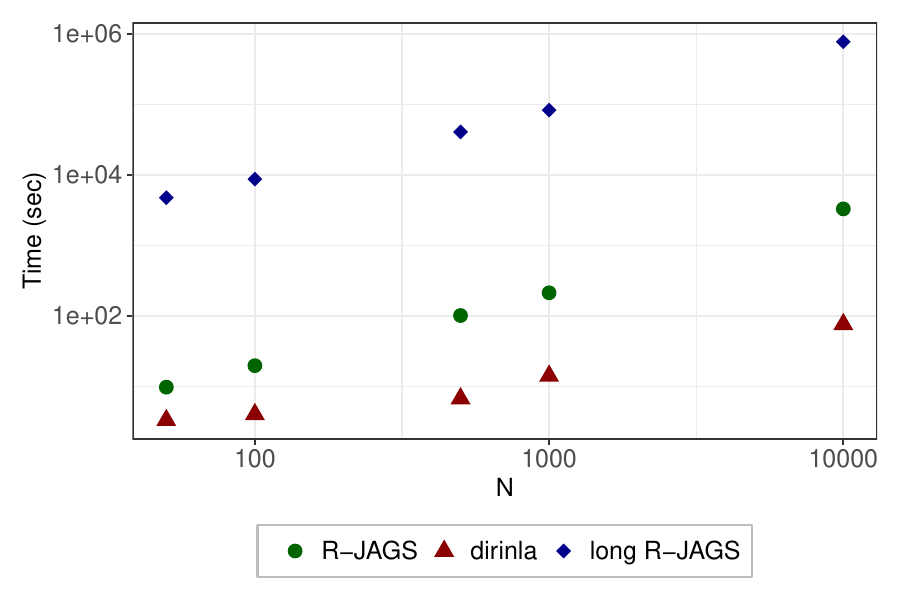} \\
\caption{Simulation 2: computational time in seconds for the different simulated data with N = 50, 100, 500, 1000 and 10000, and with the different methodologies: \pkg{R-JAGS}, \pkg{dirinla} and \pkg{long R-JAGS}.}
\label{fig:times_sim2}
}\end{center}}
\end{figure}
\begin{figure}[H]{
\begin{center}{\azul
	\includegraphics[width = \textwidth]{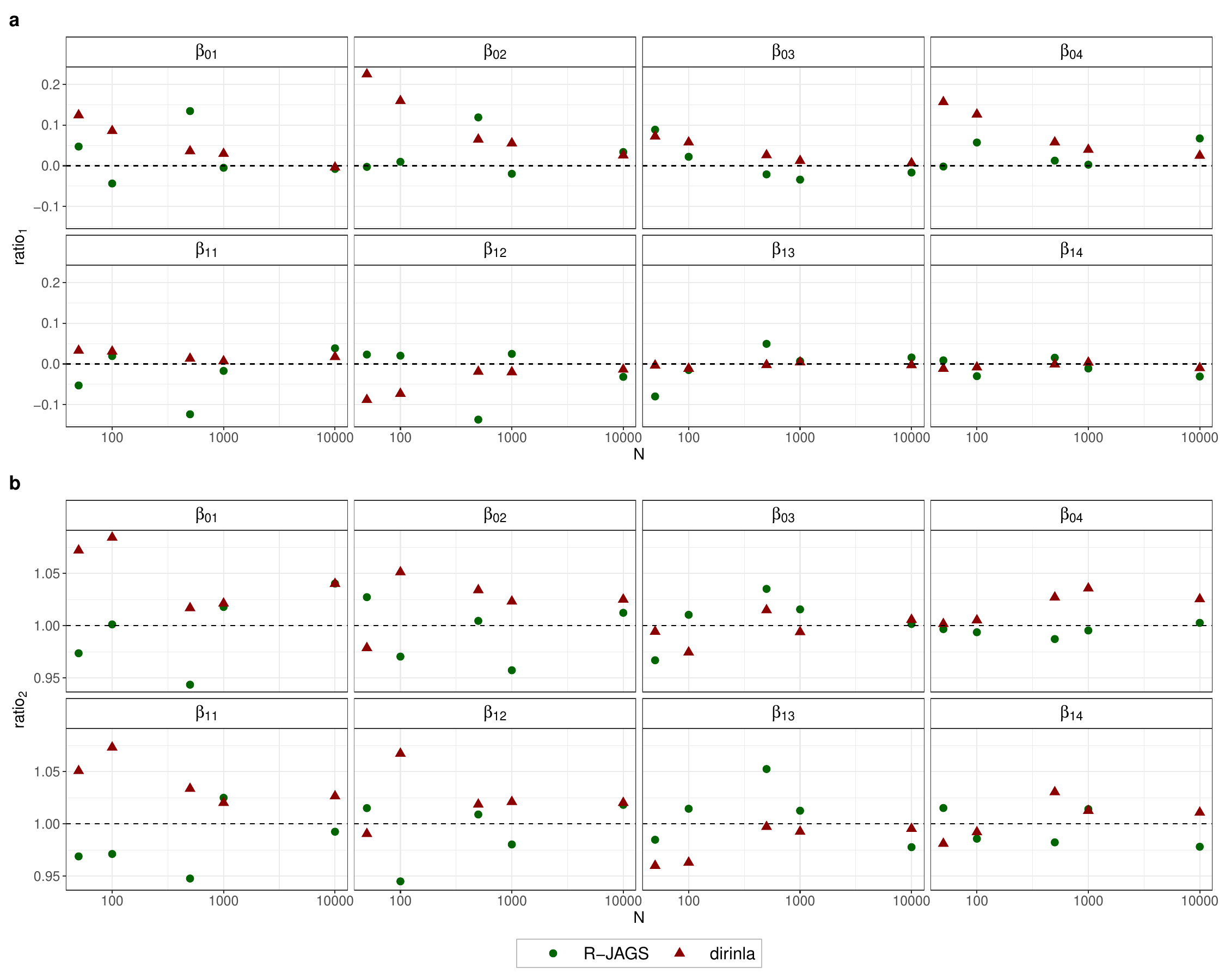} \\
\caption{Simulation 2: comparing accuracy between \pkg{dirinla} and \pkg{R-JAGS} computing two different measures: $ratio_1$ (a) and $ratio_2$ (b) for $\beta_{0c}$ and $\beta_{1c}$, $c = 1, \ldots, 4$.}
\label{fig:ratios_sim2}
}\end{center}}
\end{figure}

\subsection{Simulation 3}
The third setting is based on a Dirichlet regression with a different covariate per category without intercept and adding two shared independent random effects.
\begin{linenomath}
\begin{align}
  \vekey{Y}_{\bullet n} & \sim  \text{Dirichlet}(\alpha_{1n}, \ldots, \alpha_{4n}) \,, n = 1, \ldots, N, \nonumber \\
  \log(\alpha_{1n}) & =   \beta_{11} v_{1n} + \omega_{1i_n}, \nonumber \\
  \log(\alpha_{2n}) & =   \beta_{12} v_{2n} + \omega_{1i_n}, \nonumber \\
  \log(\alpha_{3n}) & =   \beta_{13} v_{3n} + \omega_{2i_n},  \\
  \log(\alpha_{4n}) & =   \beta_{14} v_{4n} + \omega_{2i_n}. \nonumber
\end{align}
\end{linenomath}
We simulated four different datasets of sizes $N= 50, 100$, $500, 1000$. We set values for $\beta_{1c}$ for $c = 1, \ldots, 4$ to $-1.5, 2, 1, -3 $ respectively, and we simulated covariates from a Uniform distribution on the interval $(-1,1)$. Random effects $\ve{\omega}_1$ and $\ve{\omega}_2$ were simulated from Gaussian distributions with mean 0 and {\azul standard deviations $\sigma_1 = 1/2$ (precision $\tau_1 = 4$) and $\sigma_2 = 1/3$ (precision $\tau_2 = 9$)} varying the levels of the factor ($I$), in particular, {\azul they were set to $I = 2, 5, 10, 25$}. The $i_n$ sub-index assigns each individual $n$ to a level of the factor.

As we are in the context of Bayesian LGMs, we established Gaussian prior distributions for the latent field, in this case, formed by the parameters corresponding to the fixed effects and the random effects. In particular, we assigned Gaussian prior distributions with mean 0 and precision $0.0001$ to \{$\beta_{1c}, c = 1,\ldots,4$\}, and Gaussian priors distribution with mean $0$ and precisions $\tau_1$ and $\tau_2$ for the two shared random effects $\ve{\omega}_{1}$ and $\ve{\omega}_{2}$.
Two types of priors for the $\tau_1$ and $\tau_2$ parameters were employed. Half-Gaussian priors with location $0$ and precision parameter $1$ were used for \pkg{R-JAGS}, \pkg{long R-JAGS} and \pkg{dirinla}. For \pkg{dirinla}, and additional model with a PC-prior(1, 0.01) \citep{simpson_penalising_2017} was also used. The generated data did not contain zeros and ones, so we did not use any transformation. As these models have two hyperparameters, we increased the number of iterations of the \pkg{R-JAGS} method to $20000$ and {\azul burnin} to $2000$ {\azul to achieve a given effective sample size} of the MCMC method.

In Figures \ref{fig:posteriors_100_sim3_slopes} and \ref{fig:posteriors_100_sim3_sigma} (see Simulation 3 in the supplementary material, \url{https://jmartinez-minaya.github.io/supplementary/CODA/tests.html} for the complete simulation summaries), we display the posterior distribution of the parameters and the hyperparameters with 2 levels in the factors obtained with the different methods and the different priors. In the case of the precision parameters, we also depict the priors used for them: PC-prior (pc) and Half Gaussian (hn). We observe that the shapes of the posterior densities are similar in all the cases.
\begin{figure}[H]
\begin{center}
\begin{tabular}{c}
\includegraphics[width = \textwidth]{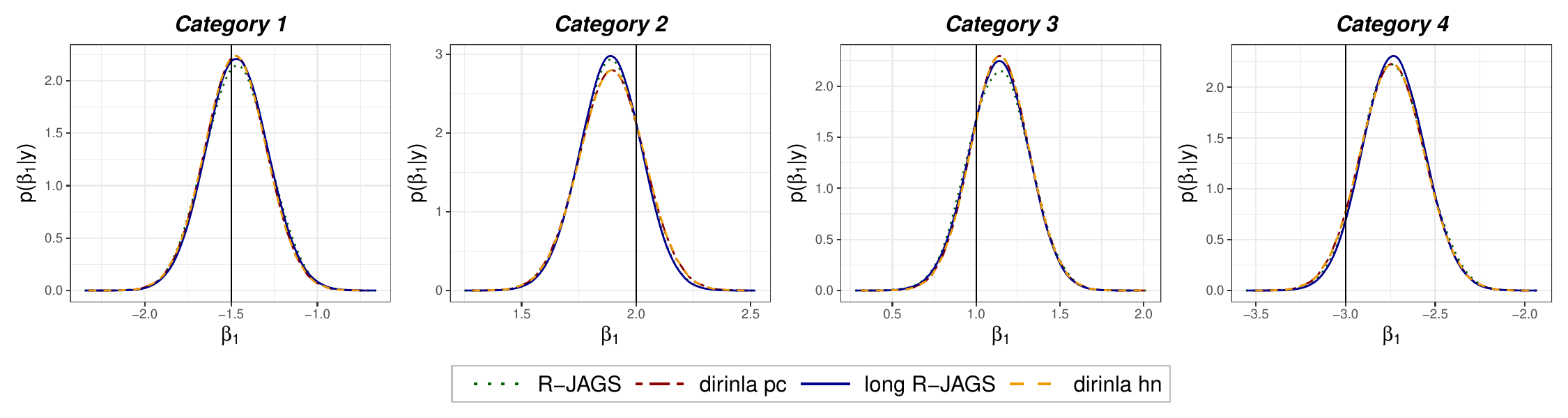} \\
\end{tabular}
\caption{Simulation 3: marginal posterior distributions of the parameters corresponding to the fixed effects using different methodologies \pkg{R-JAGS}, \pkg{dirinla} with different priors and \pkg{long R-JAGS}, when the sample size is 100 and the levels of the factor ($I$) are 2. {\azul Black vertical lines represent real values.}}
\label{fig:posteriors_100_sim3_slopes}
\end{center}
\end{figure}
\begin{figure}[H]
\begin{center}
\begin{tabular}{c}
\includegraphics[width = \textwidth]{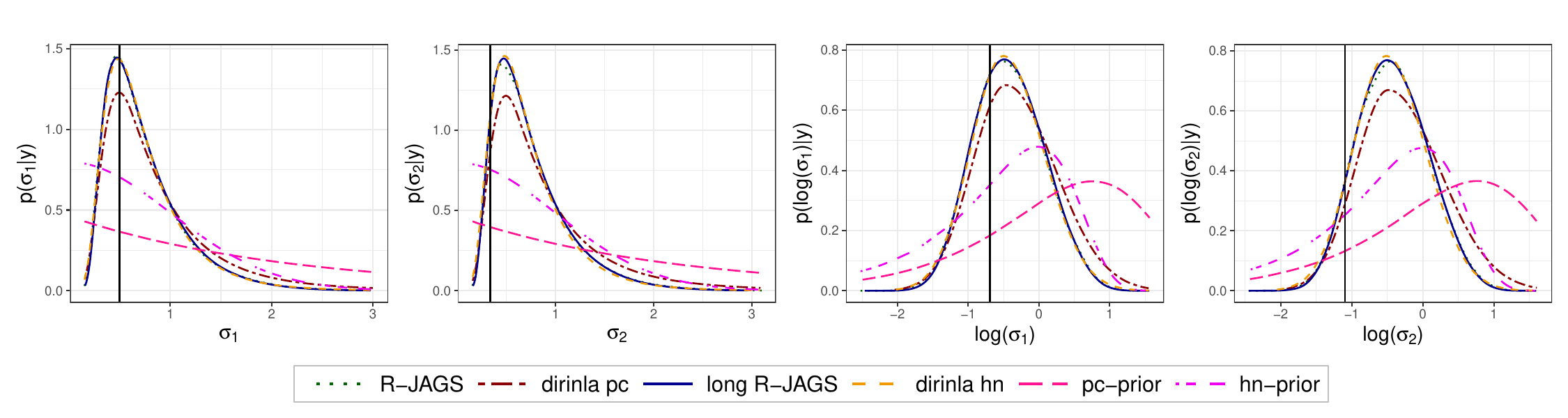} \\
\end{tabular}
\caption{Simulation 3: marginal posterior distributions of the hyperparameters with their logarithmic transformations using different methodologies \pkg{R-JAGS}, \pkg{dirinla} and \pkg{long R-JAGS}, when the sample size is 100 and the levels of the factor ($I$) are 2. Priors are also depicted in the plot. {\azul Black vertical lines represent real values.}}
\label{fig:posteriors_100_sim3_sigma}
\end{center}
\end{figure}
 {\azul Figure \ref{fig:times_simr}} shows that \pkg{dirinla} provides higher computational efficiency. And, in view of Figure \ref{fig:ratios_simr}, $ratio_1$-\pkg{dirinla} and $ratio_1$-\pkg{R-JAGS} are both closed to 0. Similar is the behaviour of $ratio_2$, in both cases, it is closed to 1, proving that when we include random effects in the model, we can also get similar approximations by decreasing the computational cost.
\begin{figure}[H]{
\begin{center}{\azul
	\includegraphics[width = 0.8\textwidth]{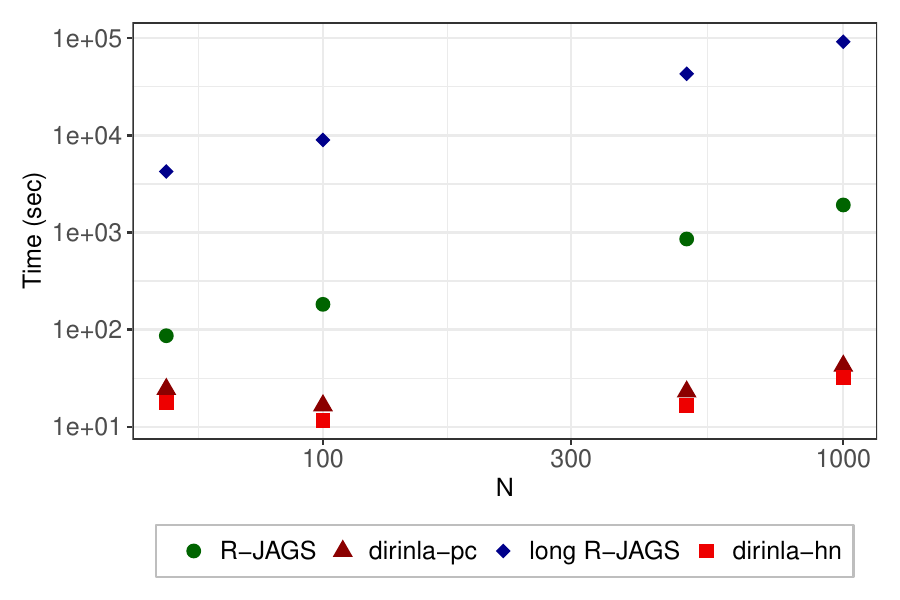} \\
\caption{Simulation 3: computational time in seconds for the different simulated data with N = 50, 100, 500 and 1000, and with the different methodologies: \pkg{R-JAGS}, \pkg{dirinla} and \pkg{long R-JAGS} when the levels of the factor ($I$) are 2.}
\label{fig:times_simr}
}\end{center}}
\end{figure}
\begin{figure}[H]{
\begin{center}{\azul
	\includegraphics[width = \textwidth]{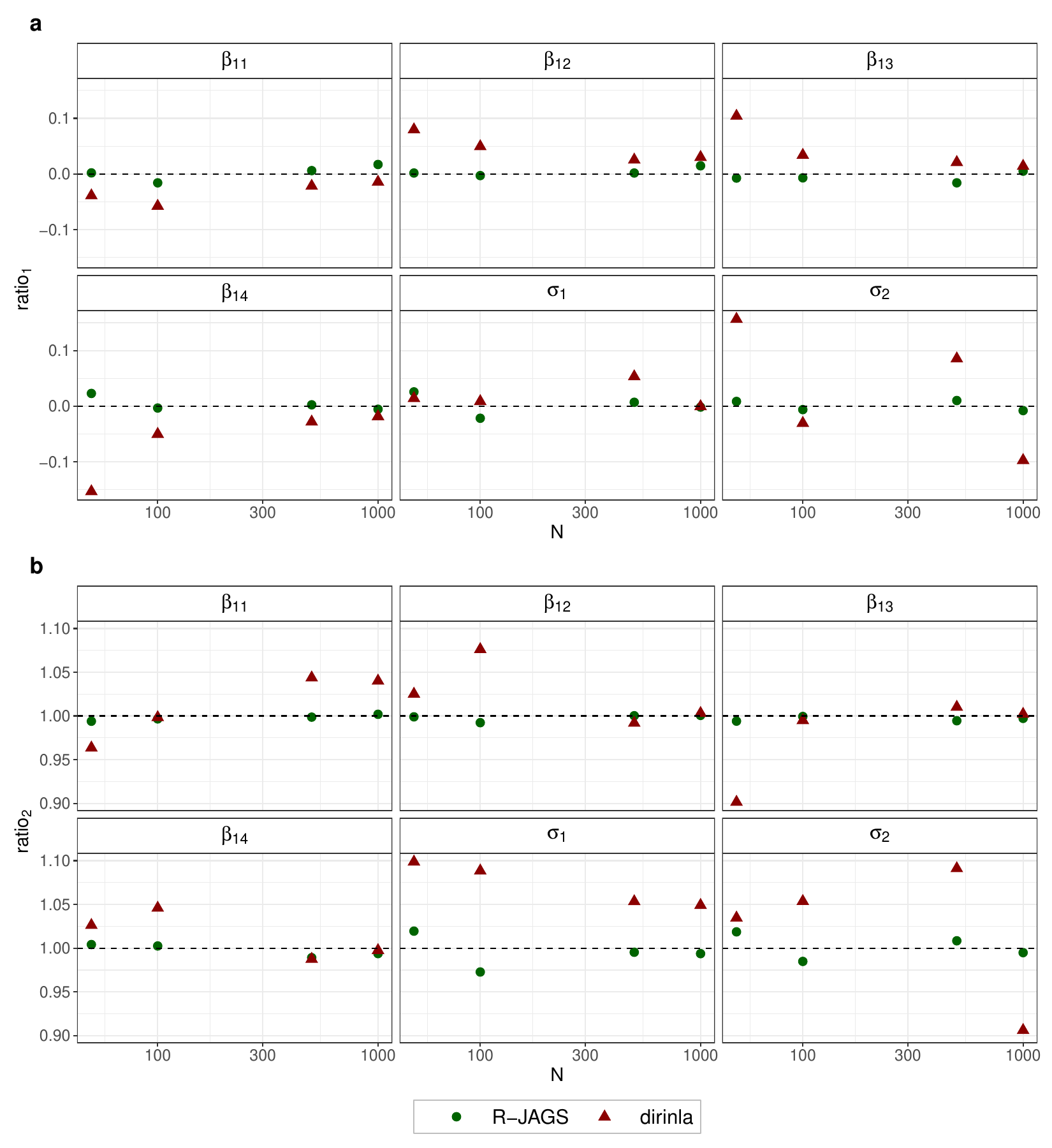} \\
\caption{Simulation 3: comparing accuracy between \pkg{dirinla} and \pkg{R-JAGS} when the prior used is Half Gaussian computing two different measures: $ratio_1$ (a) and $ratio_2$ (b) for $\beta_{1c}$, $c = 1, \ldots, 4$, and $\sigma_1$ and $\sigma_2$, when the levels of the factor ($I$) are 2.}
\label{fig:ratios_simr}
}\end{center}}
\end{figure}
If we look at the rest of the simulations conducted (See Simulation 3 in the supplementary material, \url{https://jmartinez-minaya.github.io/supplementary/CODA/tests.html}), we observe that large $N$ and small $J$ shows the influence of the different priors. And large $J$ display the difficulties of \pkg{dirinla} in approaching the posterior of the hyperparameters.

\subsection{Simulation 4}
This last simulation is based on a Dirichlet regression with just one parameter per category and without covariates as in Simulation 1. The idea is that for a fixed data size, we increase the number of categories and then simulate from all those models and compare inferential methods. In particular, {\azul we used $N = 100$ and $C = 5, 10, 15, 20$ and $30$}. We also employed \pkg{R-JAGS}, \pkg{dirinla} with the same configuration as in simulations 1 and 2. However, in the case of \pkg{long R-JAGS}, iterations were reduced to 100000 and {\azul burnin} to 10000. 

After comparing computational times needed {\azul(Figure \ref{fig:times_sim4})}, we also plot (see Figure \ref{fig:boxplot_dimension}) the measures $ratio_1$ and $ratio_2$ computed for the Dirichlet means $\mu_c$, $c = 1,\ldots, C$ to see how accurate is our method in comparison with \pkg{R-JAGS}. Once more, \pkg{dirinla} shows a higher computational speed than \pkg{R-JAGS} and a similar accuracy.
\begin{figure}[H]{
\begin{center}{\azul
	\includegraphics[width = 0.8\textwidth]{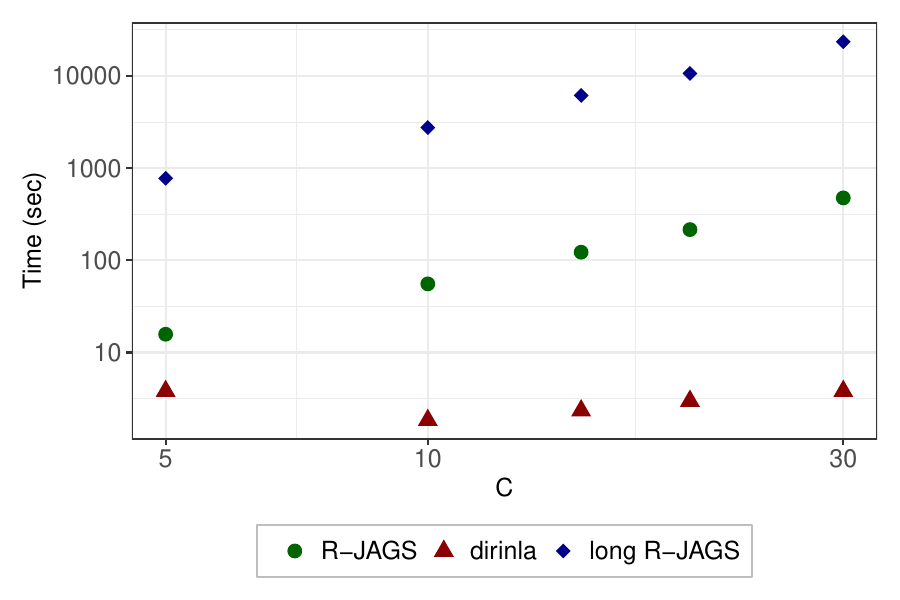} \\
\caption{Simulation 4: computational time in seconds for the different simulated data with $C = 5, 10, 15, 20$ and $30$, $N = 100$, and with the different methodologies: \pkg{R-JAGS}, \pkg{dirinla} and \pkg{long R-JAGS}.}
\label{fig:times_sim4}
}\end{center}}
\end{figure}
\begin{figure}[H]
\begin{center}
\begin{tabular}{c}
	\includegraphics[width = 0.9\textwidth]{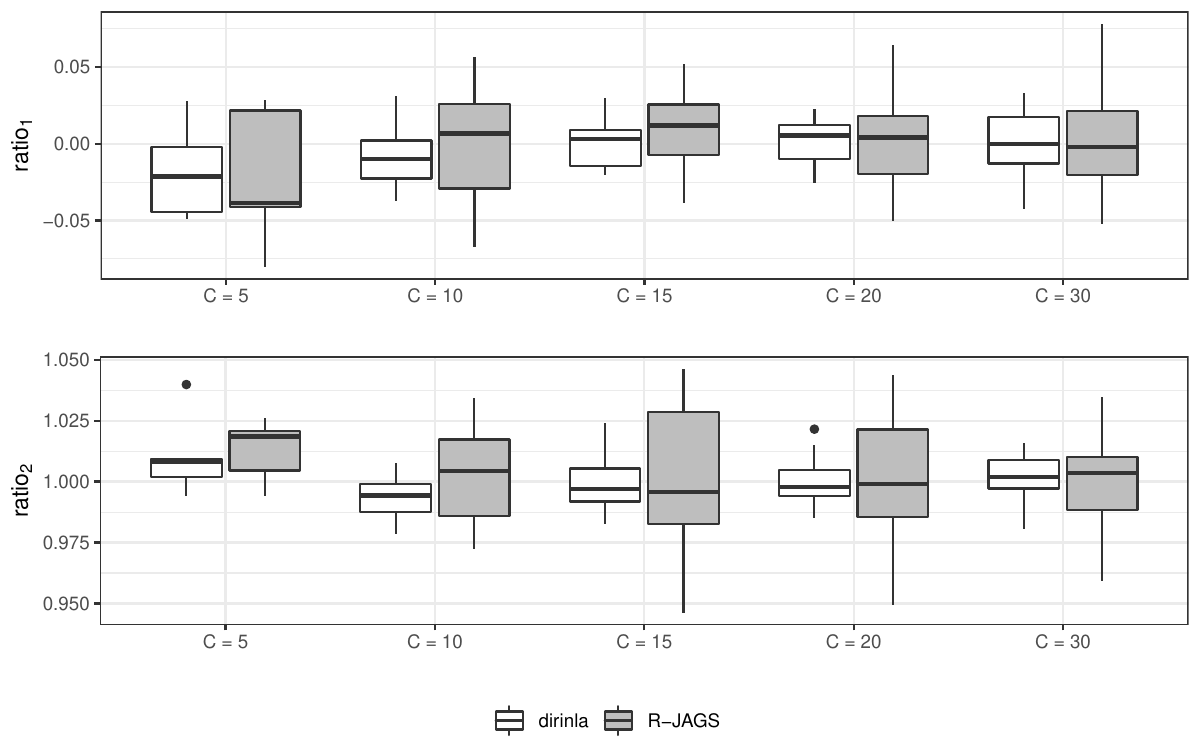} \\
\end{tabular}
\caption{Simulation 4: Boxplot comparing the $ratio_1$ and $ratio_2$ for \pkg{dirinla} and \pkg{R-JAGS} computed for the Dirichlet $\mu_c$, $c = 1,\ldots, C$. In this case, just intercepts were added to the model. The sample size was 100.}
\label{fig:boxplot_dimension}
\end{center}
\end{figure}

{\azul 
\section{Real example: \textit{Arabidopsis thaliana}} \label{sec:real_example}
After validating the use of the package and the approximation in simulated examples, this Section shows an application of the INLA approach for Dirichlet regression models in a real setting. In particular, we worked with a collection of 301 accessions of the annual plant \textit{Arabidopsis thaliana} in the Iberian Peninsula using the four genetic  clusters (gc) infered in \cite{martinez-minaya2019}: gc1, gc2, gc3 and gc4; categories whose values are summing up to one. Here, we depict how these four gcs can be modeled using Dirichlet regression in terms of two covariates: the annual mean temperature (BIO1) and the annual precipitation (BIO12) \citep{martinez-minaya2019}. The aim is explain the distribution of each gc using those two climatic covariates. The complete dataset was downloaded from the repository \url{https://zenodo.org/record/2552025#.YtbLgHZByUl}.

The Dirichlet regression model that relates the proportion of each gc, i.e., the multivariate response $\ve{Y}_{\bullet n}$, to the covariates is
\begin{linenomath}
\begin{align}
  \vekey{Y}_{\bullet n} & \sim  \text{Dirichlet}(\alpha_{1n}, \ldots, \alpha_{4n}) \,, n = 1, \ldots, 301, \nonumber \\
  \log(\alpha_{1n}) & =  \beta_{01} + \beta_{11} BIO1_{n} + \beta_{21} BIO12_{n}, \nonumber \\
  \log(\alpha_{2n}) & =  \beta_{02} + \beta_{12} BIO1_{n} + \beta_{22} BIO12_{n}, \nonumber \\
  \log(\alpha_{3n}) & =  \beta_{03} + \beta_{13} BIO1_{n} + \beta_{23} BIO12_{n},  \\
  \log(\alpha_{4n}) & =  \beta_{04} + \beta_{14} BIO1_{n} + \beta_{24} BIO12_{n}\,. \nonumber
\end{align}
\end{linenomath}
Remark that due to the different scale of the covariates, we have scaled both to have mean 0 and standard deviation 1. We assigned vague prior distributions for the latent field, in particular $p(x_m) \sim$ $\mathcal{N}(0, \tau = 0.01)$, $m = 1, \ldots, 12$. As the data did not include zeros and ones, we did not do any transformation}. In a similar approach as in the previous Section, the number of iterations used in \pkg{R-JAGS} were $20000$ with a {\azul burnin} of $2000$, a thinning of $5$ and $3$ chains, while in the case of the \pkg{long R-JAGS} we used $1000000$ of iterations with a {\azul burnin} of $100000$, a thinning of $5$ and $3$ chains.

{\azul Figure \ref{fig:posteriors_real} displays the marginal posterior distributions for $\beta_{0c}$, $\beta_{1c}$ and $\beta_{2c}$, $c=1, \ldots, 4$. In most cases distributions obtained with \pkg{R-JAGS} perfectly match with those obtained using \pkg{dirinla}. \pkg{R-JAGS} took $405.152$ seconds, \pkg{dirinla} $8.517$ seconds and \pkg{long R-JAGS} $9886.677$.}
\begin{figure}[H] {
\begin{center}
	\includegraphics[width = \textwidth]{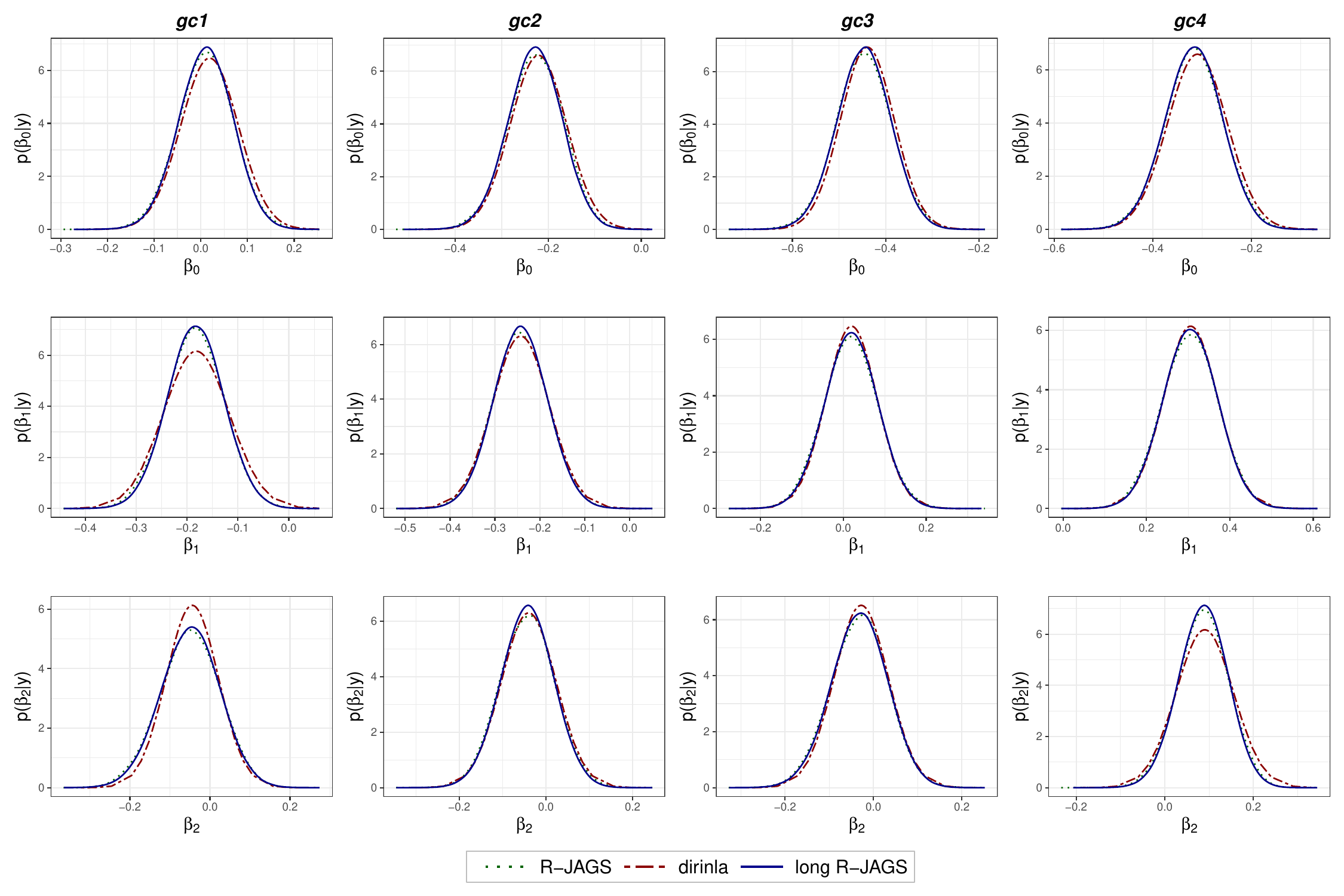}
\caption{Marginal posterior distributions of the parameters corresponding to the fixed effects using three different methodologies: \pkg{R-JAGS}, \pkg{dirinla} and \pkg{long R-JAGS}}
\label{fig:posteriors_real}
\end{center}}
\end{figure}
{
As in the previous examples, we also computed the measures $ratio_1$ and $ratio_2$ for \pkg{dirinla} and \pkg{R-JAGS}. Results are depicted in Table \ref{table:comp_ratios_real} and, we can observe that approaches done for \pkg{dirinla} seems to have a good behaviour in comparison with \pkg{R-JAGS} not only in accuracy but also in computational cost.
\begin{table}[H]{{\azul
\begin{center} {
\caption{Comparing accuracy between \pkg{dirinla} and \pkg{R-JAGS} computing $ratio_1$ and $ratio_2$ for $\beta_{0c}$, $\beta_{1c}$ and $\beta_{2c}$, $c = 1, \ldots, 4$.}
\begin{tabular}{r  rr  rr}
  \hline
  & \multicolumn{2}{c}{$ratio_1$} & \multicolumn{2}{c}{$ratio_2$} \\
 & \pkg{dirinla} & \pkg{R-JAGS} & \pkg{dirinla} & \pkg{R-JAGS} \\ 
   \hline
  $\beta_{01}$ & 0.2255 & -0.0020 & 1.0594 & 1.0026 \\ 
  $\beta_{02}$ & 0.2279 & 0.0237 & 1.0344 & 1.0013 \\ 
  $\beta_{03}$ & 0.2274 & -0.0071 & 0.9957 & 1.0042 \\ 
  $\beta_{04}$ & 0.2197 & 0.0075 & 1.0575 & 1.0003 \\
  \hline 
  $\beta_{11}$ & -0.0056 & -0.0219 & 1.0886 & 0.9967 \\ 
  $\beta_{12}$ & -0.0116 & -0.0180 & 1.0322 & 0.9977 \\ 
  $\beta_{13}$ & -0.0108 & -0.0099 & 0.9549 & 1.0023 \\ 
  $\beta_{14}$ & 0.0011 & -0.0097 & 1.0586 & 1.0011 \\ 
  \hline
  $\beta_{21}$ & 0.0509 & -0.0115 & 0.9318 & 0.9945 \\ 
  $\beta_{22}$ & 0.0560 & -0.0006 & 0.9499 & 1.0097 \\ 
  $\beta_{23}$ & 0.0836 & -0.0137 & 0.9996 & 0.9890 \\ 
  $\beta_{24}$ & 0.0612 & -0.0042 & 1.1022 & 1.0110 \\ 
  \hline
\end{tabular}
\label{table:comp_ratios_real}}
\end{center}}}
\end{table}

\section{Concluding remarks and future work} \label{sec:conclusions}
In this paper the INLA methodology is extended to fit models with Dirichlet response. In particular, we present both the calculations and a package to make inference and prediction for Dirichlet regression. The main idea underneath the proposed method is to approximate the multivariate likelihoods with univariate ones that can be fitted by \pkg{R-INLA}, in particular, Gaussian likelihoods. This idea is similar to the one proposed for modeling Multinomial likelihood in \pkg{R-INLA}, where using the Poisson trick \citep{baker1994} to reparametrize the model we just need to fit independent Poisson observations. \cite{simpson2016} use a similar strategy, constructing a Poisson approximation to the true log-Gaussian Cox process likelihood and enabling to make inference on a regular lattice over the observation window by counting the number of points in each cell. This technique has been already implemented in the \proglang{R} package \pkg{inlabru} \citep{inlabru2018,bachl2019}.

It is widely known that there exist a close connection between the Bayesian integrated likelihood and an adjusted profile likelihood \citep{lee2018}. In particular, the integrated likelihood is approximately the adjusted profile likelihood in the case of orthogonal parameters \citep{cox1987}, we just need to use the Laplace integral approximation to proof it. For the case of Dirichlet response that we deal with in this paper, we propose a quadratic approximation of the likelihood,  allowing us to measure the effect of the likelihood in the posterior. However, for the case of the adjusted profile likelihood, we have not found a similar approximation for the Dirichlet regression. Maybe, the method {\azul we have developed} can help to do approximations using the adjusted profile likelihood of Dirichlet response.

Finally, all examples presented here are focused on models that include fixed and i.i.d.\ random effects. As we are converting the multivariate observations to conditionally independent Gaussian observations that only depend on the linear predictor, and not directly on the individual latent variables, it should be expected that more complex random effects can be incorporated in the model structures. In particular, the \pkg{R-INLA} method itself uses the same hierarchical model building technique to separate general complex structured random effect components (including spatial, temporal, and spatio-temporal models) from the observation models, that are only linked to the combined linear predictor values. For such more general models, the computational scaling cost benefits in comparison with \pkg{R-JAGS} should be even more pronounced, due to the sparse matrix algebra taken advantage of by \inla{}.


\bigskip
\begin{center}
{\large\bf SUPPLEMENTARY MATERIAL}
\end{center}

\begin{description}
	\item[R-package \pkg{dirinla}:] \proglang{R}-package \pkg{dirinla} containing code to perform the Dirichlet regression using \inla{} described in the article. The package also contains the necessary code to reproduce the real and simulation examples presented in the article. {\azul It can be installed from https://github.com/inlabru-org/dirinla}.
	\item[Simulations:] the document test.html (also in \url{https://jmartinez-minaya.github.io/supplementary/CODA/tests.html}) contains a complete scenario of simulations by checking the performance of the \pkg{dirinla} R-package.
\end{description}

\bibliographystyle{Chicago}

{\azul
\bibliography{state_of_the_art}

\begin{thebibliography}{}

\bibitem[\protect\citeauthoryear{Aitchison}{Aitchison}{1981}]{aitchison1981}
Aitchison, J. (1981).
\newblock {A New Approach to Null Correlations of Proportions}.
\newblock {\em Mathematical Geology\/}~{\em 13\/}(2), 175--189.

\bibitem[\protect\citeauthoryear{Aitchison}{Aitchison}{1982}]{aitchison1982}
Aitchison, J. (1982).
\newblock {The Statistical Analysis of Compositional Data}.
\newblock {\em Journal of the Royal Statistical Society. Series B
  (Methodological)\/}, 139--177.

\bibitem[\protect\citeauthoryear{Aitchison}{Aitchison}{1983}]{aitchison1983}
Aitchison, J. (1983).
\newblock {Principal Component Analysis of Compositional Data}.
\newblock {\em Biometrika\/}~{\em 70\/}(1), 57--65.

\bibitem[\protect\citeauthoryear{Aitchison}{Aitchison}{1984}]{aitchison1984}
Aitchison, J. (1984).
\newblock {The Statistical Analysis of Geochemical Compositions}.
\newblock {\em Journal of the International Association for Mathematical
  Geology\/}~{\em 16\/}(6), 531--564.

\bibitem[\protect\citeauthoryear{Aitchison}{Aitchison}{1986}]{aitchison1986}
Aitchison, J. (1986).
\newblock {\em {The statistical Analysis of Compositional Data}}.
\newblock Chapman and Hall London.

\bibitem[\protect\citeauthoryear{Aitchison and Egozcue}{Aitchison and
  Egozcue}{2005}]{aitchison2005}
Aitchison, J. and J.~J. Egozcue (2005).
\newblock {Compositional Data Analysis: Where Are We and Where Should We Be
  Heading?}
\newblock {\em Mathematical Geology\/}~{\em 37\/}(7), 829--850.

\bibitem[\protect\citeauthoryear{Anderson}{Anderson}{1958}]{anderson1958}
Anderson, T.~W. (1958).
\newblock {\em {An Introduction to Multivariate Statistical Analysis}},
  Volume~2.
\newblock Wiley New York.

\bibitem[\protect\citeauthoryear{Bachl and Lindgren}{Bachl and
  Lindgren}{2018}]{inlabru2018}
Bachl, F.~E. and F.~Lindgren (2018).
\newblock {\em {inlabru: Spatial Inference using Integrated Nested Laplace
  Approximation}}.
\newblock R package version 2.1.9.

\bibitem[\protect\citeauthoryear{Bachl, Lindgren, Borchers, and Illian}{Bachl
  et~al.}{2019}]{bachl2019}
Bachl, F.~E., F.~Lindgren, D.~L. Borchers, and J.~B. Illian (2019).
\newblock {inlabru: an R Package for Bayesian Spatial Modelling from Ecological
  Survey Data}.
\newblock {\em Methods in Ecology and Evolution\/}~{\em 10}, 760--766.

\bibitem[\protect\citeauthoryear{Baker}{Baker}{1994}]{baker1994}
Baker, S.~G. (1994).
\newblock {The Multinomial-Poisson Transformation}.
\newblock {\em Journal of the Royal Statistical Society: Series D (The
  Statistician)\/}~{\em 43\/}(4), 495--504.

\bibitem[\protect\citeauthoryear{Barndorff-Nielsen and Cox}{Barndorff-Nielsen
  and Cox}{1989}]{barndorff1989}
Barndorff-Nielsen, O. and D.~Cox (1989).
\newblock {\em {Asymptotic Techniques for Use in Statistics}}.
\newblock Boca Raton, FL: Chapman and Hall/CRC.

\bibitem[\protect\citeauthoryear{Blangiardo and Cameletti}{Blangiardo and
  Cameletti}{2015}]{blangiardo2015}
Blangiardo, M. and M.~Cameletti (2015).
\newblock {\em {Spatial and Spatio-Temporal {B}ayesian Models with
  {\texttt{R-INLA}}}}.
\newblock John Wiley \& Sons.

\bibitem[\protect\citeauthoryear{Bonat}{Bonat}{2018}]{wagner2018}
Bonat, W.~H. (2018).
\newblock {Multiple Response Variables Regression Models in {R}: The {mcglm}
  Package}.
\newblock {\em Journal of Statistical Software\/}~{\em 84\/}(4), 1--30.

\bibitem[\protect\citeauthoryear{Buccianti and Grunsky}{Buccianti and
  Grunsky}{2014}]{buccianti2014}
Buccianti, A. and E.~Grunsky (2014).
\newblock {Compositional Data Analysis in Geochemistry: Are we Sure to See what
  Really Occurs during Natural Processes?}
\newblock {\em Journal of Geochemical Exploration\/}~{\em 141}, 1--5.

\bibitem[\protect\citeauthoryear{Cox and Reid}{Cox and Reid}{1987}]{cox1987}
Cox, D.~R. and N.~Reid (1987).
\newblock {Parameter Orthogonality and Approximate Conditional Inference}.
\newblock {\em Journal of the Royal Statistical Society: Series B
  (Methodological)\/}~{\em 49\/}(1), 1--18.

\bibitem[\protect\citeauthoryear{Cribari-Neto and Zeileis}{Cribari-Neto and
  Zeileis}{2010}]{cribari2010}
Cribari-Neto, F. and A.~Zeileis (2010).
\newblock {Beta Regression in R}.
\newblock {\em Journal of Statistical Software\/}~{\em 34\/}(2), 1--24.

\bibitem[\protect\citeauthoryear{Dennis and Mor{\'e}}{Dennis and
  Mor{\'e}}{1977}]{dennis1977}
Dennis, Jr, J.~E. and J.~J. Mor{\'e} (1977).
\newblock {Quasi-Newton Methods, Motivation and Theory}.
\newblock {\em SIAM review\/}~{\em 19\/}(1), 46--89.

\bibitem[\protect\citeauthoryear{Douma and Weedon}{Douma and
  Weedon}{2019}]{douma2019}
Douma, J.~C. and J.~T. Weedon (2019).
\newblock {Analysing Continuous Proportions in Ecology and Evolution: A
  Practical Introduction to Beta and Dirichlet Regression}.
\newblock {\em Methods in Ecology and Evolution\/}~{\em 10\/}(9), 1412--1430.

\bibitem[\protect\citeauthoryear{Dumuid, Stanford, Martin-Fern{\'a}ndez,
  Pedi{\v{s}}i{\'c}, Maher, Lewis, Hron, Katzmarzyk, Chaput, Fogelholm,
  et~al.}{Dumuid et~al.}{2018}]{dumuid2018}
Dumuid, D., T.~E. Stanford, J.-A. Martin-Fern{\'a}ndez,
  {\v{Z}}.~Pedi{\v{s}}i{\'c}, C.~A. Maher, L.~K. Lewis, K.~Hron, P.~T.
  Katzmarzyk, J.-P. Chaput, M.~Fogelholm, et~al. (2018).
\newblock {Compositional Data Analysis for Physical Activity, Sedentary Time
  and Sleep Research}.
\newblock {\em Statistical Methods in Medical Research\/}~{\em 27\/}(12),
  3726--3738.

\bibitem[\protect\citeauthoryear{Engle and Rowan}{Engle and
  Rowan}{2014}]{engle2014}
Engle, M.~A. and E.~L. Rowan (2014).
\newblock {Geochemical Evolution of Produced Waters from Hydraulic Fracturing
  of the Marcellus Shale, Northern Appalachian Basin: A Multivariate
  Compositional Data Analysis Approach}.
\newblock {\em International Journal of Coal Geology\/}~{\em 126}, 45--56.

\bibitem[\protect\citeauthoryear{Fairclough, Dumuid, Mackintosh, Stone, Dagger,
  Stratton, Davies, and Boddy}{Fairclough et~al.}{2018}]{fairclough2018}
Fairclough, S.~J., D.~Dumuid, K.~A. Mackintosh, G.~Stone, R.~Dagger,
  G.~Stratton, I.~Davies, and L.~M. Boddy (2018).
\newblock {Adiposity, Fitness, Health-Related Quality of Life and the
  Reallocation of Time between Children's School Day Activity Behaviours: A
  Compositional Data Analysis}.
\newblock {\em Preventive Medicine Reports\/}~{\em 11}, 254--261.

\bibitem[\protect\citeauthoryear{Gelman, Hwang, and Vehtari}{Gelman
  et~al.}{2014}]{gelman2014}
Gelman, A., J.~Hwang, and A.~Vehtari (2014).
\newblock {Understanding Predictive Information Criteria for Bayesian Models}.
\newblock {\em Statistics and Computing\/}~{\em 24\/}(6), 997--1016.

\bibitem[\protect\citeauthoryear{Gneiting and Raftery}{Gneiting and
  Raftery}{2007}]{gneiting2007}
Gneiting, T. and A.~E. Raftery (2007).
\newblock {Strictly Proper Scoring Rules, Prediction, and Estimation}.
\newblock {\em Journal of the American Statistical Association\/}~{\em
  102\/}(477), 359--378.

\bibitem[\protect\citeauthoryear{G{\'o}mez-Rubio}{G{\'o}mez-Rubio}{2020}]{gomez-rubio2020}
G{\'o}mez-Rubio, V. (2020).
\newblock {\em {Bayesian Inference with INLA}}.
\newblock CRC Press.

\bibitem[\protect\citeauthoryear{Hadfield}{Hadfield}{2010}]{hadfield2010}
Hadfield, J.~D. (2010).
\newblock {MCMC Methods for Multi-Response Generalized Linear Mixed Models: The
  {MCMCglmm} {R} Package}.
\newblock {\em Journal of Statistical Software\/}~{\em 33\/}(2), 1--22.

\bibitem[\protect\citeauthoryear{Hijazi and Jernigan}{Hijazi and
  Jernigan}{2009}]{hijazi2009}
Hijazi, R.~H. and R.~W. Jernigan (2009).
\newblock {Modelling Compositional Data using Dirichlet Regression Models}.
\newblock {\em Journal of Applied Probability \& Statistics\/}~{\em 4\/}(1),
  77--91.

\bibitem[\protect\citeauthoryear{Klein, Kneib, Klasen, and Lang}{Klein
  et~al.}{2015}]{klein2015}
Klein, N., T.~Kneib, S.~Klasen, and S.~Lang (2015).
\newblock {Bayesian Structured Additive Distributional Regression for
  Multivariate Responses}.
\newblock {\em Journal of the Royal Statistical Society: Series C (Applied
  Statistics)\/}~{\em 64\/}(4), 569--591.

\bibitem[\protect\citeauthoryear{Kobal, Kastelec, and Eler}{Kobal
  et~al.}{2017}]{kobal2017}
Kobal, M., D.~Kastelec, and K.~Eler (2017).
\newblock {Temporal Changes of Forest Species Composition Studied by
  Compositional Data Approach}.
\newblock {\em iForest-Biogeosciences and Forestry\/}~{\em 10\/}(4), 729.

\bibitem[\protect\citeauthoryear{Krainski, G{\'o}mez-Rubio, Bakka, Lenzi,
  Castro-Camilo, Simpson, Lindgren, and Rue}{Krainski
  et~al.}{2018}]{krainski2018}
Krainski, E.~T., V.~G{\'o}mez-Rubio, H.~Bakka, A.~Lenzi, D.~Castro-Camilo,
  D.~Simpson, F.~Lindgren, and H.~Rue (2018).
\newblock {\em {Advanced Spatial Modeling with Stochastic Partial Differential
  Equations using R and INLA}}.
\newblock CRC Press.

\bibitem[\protect\citeauthoryear{Lee, Nelder, and Pawitan}{Lee
  et~al.}{2018}]{lee2018}
Lee, Y., J.~A. Nelder, and Y.~Pawitan (2018).
\newblock {\em {Generalized Linear Models with Random Effects: Unified Analysis
  via H-likelihood}}, Volume 153.
\newblock CRC Press.

\bibitem[\protect\citeauthoryear{Maier}{Maier}{2014}]{maier2014}
Maier, M.~J. (2014).
\newblock {DirichletReg: Dirichlet Regression for Compositional Data in R}.

\bibitem[\protect\citeauthoryear{Mart{\'\i}nez-Minaya, Conesa, Fortin,
  Alonso-Blanco, Pic{\'o}, and Marcer}{Mart{\'\i}nez-Minaya
  et~al.}{2019}]{martinez-minaya2019}
Mart{\'\i}nez-Minaya, J., D.~Conesa, M.-J. Fortin, C.~Alonso-Blanco, F.~X.
  Pic{\'o}, and A.~Marcer (2019).
\newblock {A hierarchical Bayesian Beta Regression Approach to Study the
  Effects of Geographical Genetic Structure and Spatial Autocorrelation on
  Species Distribution Range Shifts}.
\newblock {\em Molecular ecology resources\/}~{\em 19\/}(4), 929--943.

\bibitem[\protect\citeauthoryear{Masarotto and Varin}{Masarotto and
  Varin}{2017}]{masarotto2017}
Masarotto, G. and C.~Varin (2017).
\newblock {Gaussian Copula Regression in {R}}.
\newblock {\em Journal of Statistical Software\/}~{\em 77\/}(8), 1--26.

\bibitem[\protect\citeauthoryear{Monyai, Lesaoana, Darikwa, and
  Nyamugure}{Monyai et~al.}{2016}]{monyai2016}
Monyai, S., M.~Lesaoana, T.~Darikwa, and P.~Nyamugure (2016).
\newblock {Application of Multinomial Logistic Regression to Educational
  Factors of the 2009 General Household Survey in South Africa}.
\newblock {\em Journal of Applied Statistics\/}~{\em 43\/}(1), 128--139.

\bibitem[\protect\citeauthoryear{Moraga}{Moraga}{2019}]{moraga2019}
Moraga, P. (2019).
\newblock {\em {Geospatial Health Data: Modeling and Visualization with R-INLA
  and Shiny}}.
\newblock CRC Press.

\bibitem[\protect\citeauthoryear{Nocedal and Wright}{Nocedal and
  Wright}{2006}]{nocedal2006}
Nocedal, J. and S.~Wright (2006).
\newblock {\em {Numerical Optimization}}.
\newblock Springer Science \& Business Media.

\bibitem[\protect\citeauthoryear{Nowosad and Stepinski}{Nowosad and
  Stepinski}{2018}]{nowosad2018}
Nowosad, J. and T.~Stepinski (2018).
\newblock {Spatial Association between Regionalizations Using the
  Information-Theoretical V-Measure}.
\newblock {\em International Journal of Geographical Information
  Science\/}~{\em 32\/}(12), 2386--2401.

\bibitem[\protect\citeauthoryear{Odeyemi, Pollind, Peeler, Nozawa, Vesely,
  Page, Rakovski, and El-Askary}{Odeyemi et~al.}{2019}]{odeyemi2019}
Odeyemi, Y., M.~Pollind, R.~Peeler, K.~Nozawa, D.~Vesely, A.~Page, C.~Rakovski,
  and H.~El-Askary (2019).
\newblock {Review of Climate Research and Funding 1993\~{} 2017: A Multinomial
  Logistic Regression Approach}.
\newblock {\em Journal of Environmental Informatics Letters\/}~{\em 1\/}(2),
  94--101.

\bibitem[\protect\citeauthoryear{Piccini, Marchetti, Rivieccio, and
  Napoli}{Piccini et~al.}{2019}]{piccini2019}
Piccini, C., A.~Marchetti, R.~Rivieccio, and R.~Napoli (2019).
\newblock {Multinomial Logistic Regression with Soil Diagnostic Features and
  Land Surface Parameters for Soil Mapping of Latium (Central Italy)}.
\newblock {\em Geoderma\/}~{\em 352}, 385--394.

\bibitem[\protect\citeauthoryear{Plummer}{Plummer}{2016}]{plummer2016}
Plummer, M. (2016).
\newblock {\em {rjags: Bayesian Graphical Models using MCMC}}.
\newblock R package version 4-6.

\bibitem[\protect\citeauthoryear{{R Core Team}}{{R Core Team}}{2021}]{rmanual}
{R Core Team} (2021).
\newblock {\em {R: A Language and Environment for Statistical Computing}}.
\newblock Vienna, Austria: R Foundation for Statistical Computing.

\bibitem[\protect\citeauthoryear{Rue and Held}{Rue and Held}{2005}]{rue2005}
Rue, H. and L.~Held (2005).
\newblock {\em {Gaussian {Markov Random Fields: Theory and Applications}}}.
\newblock Chapman \& Hall.

\bibitem[\protect\citeauthoryear{Rue, Martino, and Chopin}{Rue
  et~al.}{2009}]{rue2009inla}
Rue, H., S.~Martino, and N.~Chopin (2009).
\newblock {Approximate Bayesian Inference for Latent Gaussian Models by using
  Integrated Nested Laplace Approximations}.
\newblock {\em Journal of the Royal Statistical Society: Series B (Statistical
  Methodology)\/}~{\em 71\/}(2), 319--392.

\bibitem[\protect\citeauthoryear{Rue, Riebler, S{\o}rbye, Illian, Simpson, and
  Lindgren}{Rue et~al.}{2017}]{rue2017bayesian}
Rue, H., A.~Riebler, S.~H. S{\o}rbye, J.~B. Illian, D.~P. Simpson, and F.~K.
  Lindgren (2017).
\newblock {Bayesian Computing with \texttt{INLA}: a Review}.
\newblock {\em Annual Review of Statistics and Its Application\/}~{\em 4},
  395--421.

\bibitem[\protect\citeauthoryear{Sennhenn-Reulen}{Sennhenn-Reulen}{2018}]{sennhenn2018}
Sennhenn-Reulen, H. (2018).
\newblock {Bayesian Regression for a Dirichlet Distributed Response using
  Stan}.
\newblock {\em arXiv preprint arXiv:1808.06399\/}.

\bibitem[\protect\citeauthoryear{Shi, Zhang, Li, et~al.}{Shi
  et~al.}{2016}]{shi2016}
Shi, P., A.~Zhang, H.~Li, et~al. (2016).
\newblock {Regression Analysis for Microbiome Compositional Data}.
\newblock {\em The Annals of Applied Statistics\/}~{\em 10\/}(2), 1019--1040.

\bibitem[\protect\citeauthoryear{Simpson, Illian, Lindgren, S{\o}rbye, and
  Rue}{Simpson et~al.}{2016}]{simpson2016}
Simpson, D., J.~B. Illian, F.~Lindgren, S.~H. S{\o}rbye, and H.~Rue (2016).
\newblock {Going off Grid: Computationally Efficient Inference for log-Gaussian
  Cox Processes}.
\newblock {\em Biometrika\/}~{\em 103\/}(1), 49--70.

\bibitem[\protect\citeauthoryear{Simpson, Rue, Riebler, Martins, and
  Sørbye}{Simpson et~al.}{2017}]{simpson_penalising_2017}
Simpson, D., H.~Rue, A.~Riebler, T.~G. Martins, and S.~H. Sørbye (2017,
  February).
\newblock Penalising {Model} {Component} {Complexity}: {A} {Principled},
  {Practical} {Approach} to {Constructing} {Priors}.
\newblock {\em Statistical Science\/}~{\em 32\/}(1), 1--28.
\newblock Publisher: Institute of Mathematical Statistics.

\bibitem[\protect\citeauthoryear{Smithson and Verkuilen}{Smithson and
  Verkuilen}{2006}]{smithson2006}
Smithson, M. and J.~Verkuilen (2006).
\newblock {A Better Lemon Squeezer? Maximum-likelihood Regression with
  Beta-Distributed Dependent Variables}.
\newblock {\em Psychological Methods\/}~{\em 11\/}(1), 54.

\bibitem[\protect\citeauthoryear{Spiegelhalter, Best, Carlin, and Van
  Der~Linde}{Spiegelhalter et~al.}{2002}]{spiegelhalter2002}
Spiegelhalter, D.~J., N.~G. Best, B.~P. Carlin, and A.~Van Der~Linde (2002).
\newblock {Bayesian Measures of Model Complexity and Fit}.
\newblock {\em Journal of the Royal Statistical Society: Series B (Statistical
  Methodology)\/}~{\em 64\/}(4), 583--639.

\bibitem[\protect\citeauthoryear{Tsilimigras and Fodor}{Tsilimigras and
  Fodor}{2016}]{tsilimigras2016}
Tsilimigras, M.~C. and A.~A. Fodor (2016).
\newblock {Compositional Data Analysis of the Microbiome: Fundamentals, Tools,
  and Challenges}.
\newblock {\em Annals of Epidemiology\/}~{\em 26\/}(5), 330--335.

\bibitem[\protect\citeauthoryear{Van~der Merwe}{Van~der
  Merwe}{2019}]{VanderMerwe2018}
Van~der Merwe, S. (2019).
\newblock {A Method for Bayesian Regression Modelling of Composition Data}.
\newblock {\em South African Statistical Journal\/}~{\em 53\/}(1), 55--64.

\bibitem[\protect\citeauthoryear{Wang, Ryan, and Faraway}{Wang
  et~al.}{2018}]{faraway2018}
Wang, X., Y.~Y. Ryan, and J.~J. Faraway (2018).
\newblock {\em {Bayesian Regression Modeling with INLA}}.
\newblock Chapman and Hall/CRC.

\bibitem[\protect\citeauthoryear{Washburne, Silverman, Leff, Bennett, Darcy,
  Mukherjee, Fierer, and David}{Washburne et~al.}{2017}]{washburne2017}
Washburne, A.~D., J.~D. Silverman, J.~W. Leff, D.~J. Bennett, J.~L. Darcy,
  S.~Mukherjee, N.~Fierer, and L.~A. David (2017).
\newblock {Phylogenetic Factorization of Compositional Data Yields
  Lineage-level Associations in Microbiome Datasets}.
\newblock {\em PeerJ\/}~{\em 5}, e2969.

\bibitem[\protect\citeauthoryear{Zuur, Ieno, and Saveliev}{Zuur
  et~al.}{2017}]{zuur2017beginner}
Zuur, A.~F., E.~N. Ieno, and A.~A. Saveliev (2017).
\newblock {\em {Beginner's Guide to Spatial, Temporal, and Spatial-Temporal
  Ecological Data Analysis with R-INLA}}.
\newblock Highland Statistics Ltd, Newburgh.

\end{thebibliography}
}

\newpage
\section*{Appendix}
\begin{appendix}
\appendix

\section{Equivalence between parametrizations} \label{appendix_param}
In this Appendix, we present the proof for the equivalence between the two parametrizations on Dirichlet regression models presented in the equations (\ref{eq:dirichlet_regression}) and (\ref{eq:dirichlet_regression_2}).

As mentioned in the text, we just need to construct a matrix $\ve{V}$ with dimension $N \times J$ which contains all the observed values for all the observations of all the covariates using in the different categories. Note that if the used covariates are the same in all the categories $\ve{V}$ is equal to $\ve{V}^{(c)}$. We rewrite equation (\ref{eq:dirichlet_regression}) as: {\azul
\begin{equation}\label{eq:dirichlet_regression_new}
	g(\alpha_{cn}) = \eta^*_{cn} = \vekey{V}_{n\bullet} \vekey{\tilde{\beta}}_{\bullet c} + \omega_{cn}\,\,,
\end{equation}}
being {\azul $\vekey{\tilde{\beta}}_{\bullet c}$} a $J$ column vector having $J_c$ non zero values, and being zero the rest of the elements; and equation \eqref{eq:dirichlet_regression_2}
{\azul 
\begin{linenomath}
\begin{equation}\label{eq:dirichlet_regression_2_new}
\mu_{cn} = \frac{\exp{(\eta^{*\ve{\mu}}_{cn})}}{\sum_{c = 1}^C \exp{(\eta^{*\ve{\mu}}_{cn})}} = \frac{\exp{(\vekey{V}_{n\bullet} \vekey{\tilde{\gamma}}_{\bullet c} + \omega_{cn}^{\vekey{\mu}})}}{\sum_{c = 1}^C \exp{(\vekey{V}_{n\bullet} \vekey{\tilde{\gamma}}_{\bullet c} + \omega_{cn}^{\vekey{\mu}})}}  \,\,.
\end{equation}
\end{linenomath}
being $\vekey{\tilde{\gamma}}^{c}$ a $J$ column vector having $J_c$ non zero values and being zero the rest of the elements. $\omega_{cn}^{\vekey{\mu} (c)}$ represents a realization of a random effect for this parametrization.}

In the alternative parametrization \eqref{eq:dirichlet_regression_new}, as we have $C$ categories for which the means must always sum up to $1$, we employ a Multinomial $\textit{logit}$ strategy as in Multinomial regression, where the linear predictor of one category is set to zero, whereby it is virtually omitted and becomes the reference. Let {\azul $\eta_{1n}^{*\ve{\mu}}$ be the reference, then:
\begin{equation}\label{eq:alternative_parametrisation}
  \eta_{1n}^{*\ve{\mu}} =  0 \,.\nonumber\\
\end{equation}
Rewriting the equation {\azul \eqref{eq:dirichlet_regression_2_new}} reveals that
\begin{eqnarray}\label{eq:equivalence_beta_gamma}
\mu_{cn} &=& \frac{\exp{(\vekey{V}_{n\bullet} \vekey{\tilde{\gamma}}_{\bullet c} + \omega_{cn})}}{\sum_{c=1}^C \exp{(\vekey{V}_{n\bullet} \vekey{\tilde{\gamma}}_{\bullet c} + \omega_{cn})}} =
\frac{\exp{(\vekey{V}_{n \bullet} \vekey{\tilde{\beta}}_{\bullet c} - \vekey{V}_{n \bullet} \vekey{\tilde{\beta}}_{\bullet 1} + \omega_{cn} - \omega_{1n})}}{\sum_{c=1}^C \exp{(\vekey{V}_{n \bullet} \vekey{\tilde{\beta}}_{\bullet c} - \vekey{V}_{n \bullet} \vekey{\tilde{\beta}}_{\bullet 1} + \omega_{cn} - \omega_{1n})}} \,,
\end{eqnarray}
so that
\begin{equation}
\tilde{\ve{\gamma}}_{\bullet c} = \tilde{\ve{\beta}}_{\bullet c} - \tilde{\ve{\beta}}_{\bullet 1}\,,
\end{equation}
and
\begin{equation}
\omega_{cn}^{\ve{\mu}} = \omega_{cn}- \omega_{1n}\,,
\end{equation}
as we wanted to proof.}

\section{Likelihood approximation effect} \label{append_likelihood_approx}
This Appendix presents both the proof for the Theorem \ref{theorem_laplace} and the Proposition \ref{proposition_laplace}, which is the expansion of the theorem to multiple observations.
{\azul
Let $\vekey{\eta}_{n}:=\vekey{\eta}_{\bullet n }$ denote the linear predictor corresponding to the $n$th observation $\vekey{y}_n:=\vekey{Y}_{ \bullet n}$; let also define $l(\vekey{y} \mid \vekey{x})=-\log p(\vekey{y} \mid \vekey{x})$ for any $\vekey{y}$ and $\vekey{x}$, and so, let $l(\vekey{y}_n \mid \vekey{\eta}_n)=-\log p(\vekey{y}_n\mid \vekey{\eta}_{n})$ denote the log-likelihood function expressed for the $n$th observation, being $\vekey{y}_n \in \mathbb{R}^C$ and $\vekey{\eta}_n \in \mathbb{R}^C$. Using the Taylor series expansion in vector $\vekey{\eta}^0_n$, we obtain the approximation:
\begin{linenomath}
\begin{align} \label{eq:log_likelihood_taylor_1_ap}
\lefteqn{l(\vekey{y}_n \mid \vekey{\eta}_n)  \approx}  &  \nonumber \\
& \approx  l(\vekey{y}_n \mid \vekey{\eta}^0_n) + [\nabla_{\vekey{\eta}_n}(l)(\vekey{\eta}^0_n, \vekey{y}_n)]^T [\vekey{\eta}_n - \vekey{\eta}^0_n] \nonumber \\
&\phantom{\approx }  + \frac{1}{2} [\vekey{\eta}_n - \vekey{\eta}^0_n]^T [\nabla^2_{\vekey{\eta}_n}(l)(\vekey{\eta}^0_n, \vekey{y}_n)] [\vekey{\eta}_n - \vekey{\eta}^0_n] \nonumber \\
& =  l(\vekey{y}_n \mid \vekey{\eta}^0_n) + (\vekey{g}^0_{\ve{\eta}_n})^T [\vekey{\eta}_n - \vekey{\eta}^0_n] + \frac{1}{2} [\vekey{\eta}_n - \vekey{\eta}^0_n]^T \vekey{H}^0_{\ve{\eta}_n} [\vekey{\eta}_n - \vekey{\eta}^0_n] \\
&= C_1 + \frac{1}{2}[\vekey{\eta}_n - (\vekey{\eta}^0_n -(\vekey{H}^0_{\ve{\eta}_n})^{-1} \vekey{g}^0_{\ve{\eta}_n})]^T \vekey{H}^0_{\ve{\eta}_n} [\vekey{\eta}_n - (\vekey{\eta}^0_n - (\vekey{H}^0_{\ve{\eta}_n})^{-1} \vekey{g}^0_{\ve{\eta}_n})] \,, \nonumber
\end{align}
\end{linenomath}
where $\vekey{g}^0_{\ve{\eta}_n}=\nabla_{\vekey{\eta}_n}(l)(\vekey{\eta}^0_n, \vekey{y}_n)$ and $\vekey{H}^0_{\ve{\eta}_n}$ is either the true Hessian ($\nabla^2_{\vekey{\eta}_n}(l)(\vekey{\eta}^0_n, \vekey{y}_n)$) or the expected Hessian ($\text{E}_{\vekey{y}_n \mid \vekey{\eta_n}} (\nabla^2_{\vekey{\eta}_n}(l)(\vekey{\eta}^0_n, \vekey{y}_n))$). $C_1$ is a constant whose value is $l(\vekey{y}_n \mid \vekey{\eta}^0_n) - \frac{1}{2} (\vekey{g}^0_{\ve{\eta}_n})^T (\vekey{H}^0_{\ve{\eta}_n})^{-1} \vekey{g}^0_{\ve{\eta}_n}$.}
{\azul 
Considering now the Cholesky factorization of $\vekey{H}^0_{\ve{\eta}_n}$, $\vekey{H}^0_{\ve{\eta}_n} = \vekey{L}^0_n (\vekey{L}^0_n)^T$, expression (\ref{eq:log_likelihood_taylor_1_ap}) can be rewritten as follows:
\begin{linenomath}
  \begin{align} \label{eq:log_likelihood_taylor_2}
\lefteqn{l(\vekey{y}_n \mid \vekey{\eta}_n)  \approx } \\
& \approx  C_1 + \frac{1}{2}[(\vekey{L}^0_n)^T \vekey{\eta}_n - (\vekey{L}^0_n)^T (\vekey{\eta}^0_n - (\vekey{H}^0_{\ve{\eta}_n})^{-1} \vekey{g}^0_{\ve{\eta}_n})]^T \nonumber \\
&\phantom{\approx C_1 + \frac{1}{2}}  [(\vekey{L}^0_n)^T \vekey{\eta}_n - (\vekey{L}^0_n)^T (\vekey{\eta}^0_n - (\vekey{H}^0_{\ve{\eta}_n})^{-1} \vekey{g}^0_{\ve{\eta}_n})] \nonumber  \,.
\end{align}
\end{linenomath}
Defining
\begin{linenomath}
\begin{align}\label{eq:z_var_definition_ap}
\vekey{z}^0_n &:= (\vekey{L}^0_n)^T [\vekey{\eta}^0_n - (\vekey{H}^0_{\ve{\eta}_n})^{-1} \vekey{g}^0_{\ve{\eta}_n}] = (\vekey{L}^0_n)^T \vekey{\eta}^0_n - (\vekey{L}^0_n)^{-1} \vekey{g}^0_{\ve{\eta}_n} \,,
\end{align}
\end{linenomath}
a conditionally Gaussian approximation is constructed.
\begin{linenomath}
\begin{align} \label{eq:log_likelihood_taylor_3}
\l(\vekey{y}_n\mid \vekey{\eta}_{n}) & \approx l(\ve{z}^0_n \mid \ve{\eta}_n) =
  C_1 +\frac{1}{2}[\vekey{z}^0_n-(\vekey{L}^0_n)^T \vekey{\eta}_n]^T  [\vekey{z}^0_n-(\vekey{L}^0_n)^T \vekey{\eta}_n] \,.
\end{align}
\end{linenomath}
Thus, $\vekey{z}^0_n \mid \vekey{\eta}_n \sim \mathcal{N}((\vekey{L}^0_n)^T \vekey{\eta}_n, \vekey{I}_{d})$, \textit{i.e.}, $z^0_{nc} \mid \vekey{\eta}_n \sim \mathcal{N}([(\vekey{L}^0_n)^T \vekey{\eta}_n]_c, 1)$. The observation vector $\vekey{y}_n$ has been converted into conditionally independent Gaussian pseudo-observations $\vekey{z}^0_n$. This approximation can be expanded to the $N$ observations.
}

To present a proof for Proposition (\ref{proposition_laplace}), we rewrite $l(\vekey{Y} \mid \vekey{\eta})$ for all the observations $N$ as in equation (\ref{eq:log_likelihood_taylor_2}).
\begin{linenomath}{\azul
\begin{align} \label{eq:log_likelihood_n_taylor_2}
\lefteqn{l(\vekey{Y} \mid \vekey{\eta})  \approx } \nonumber\\
& \approx  NC_1 + \frac{1}{2}\sum_{n=1}^n[(\vekey{L}^0_n)^T \vekey{\eta}_n - (\vekey{L}^0_n)^T (\vekey{\eta}^0_n - (\vekey{H}^0_{\ve{\eta}_n})^{-1} \vekey{g}^0_{\ve{\eta}_n})]^T  \\
&\phantom{\approx NC_1 + \frac{1}{2}\sum{n=1}^n }
[(\vekey{L}^0_n)^T \vekey{\eta}_n - (\vekey{L}^0_n)^T (\vekey{\eta}^0_n - (\vekey{H}^0_{\ve{\eta}_n})^{-1} \vekey{g}^0_{\ve{\eta}_n})] \nonumber \,.
\end{align}}
\end{linenomath}{\azul
Using the notation \[
  \vekey{\tilde{\eta}}^0=
    \underbrace{\begin{bmatrix}
      \vekey{\eta}^0_{\bullet 1} \\
      \vdots \\
      \vekey{\eta}^0_{\bullet N }
      \end{bmatrix}}_{CN \times 1} \,, \
  \vekey{g}^0_{\ve{\tilde{\eta}}}=
    \underbrace{\begin{bmatrix}
      \vekey{g}^0_1 \\
      \vdots \\
      \vekey{g}^0_N
      \end{bmatrix}}_{CN \times 1} \,, \
  \vekey{L}^0=
    \underbrace{\begin{bmatrix}
      \vekey{L}^0_1 & & 0   \\
      & \ddots & \\
      0 & & \vekey{L}^0_N
      \end{bmatrix}}_{CN \times CN} \,,
\]
\[
  \vekey{H}^0_{\ve{\tilde{\eta}}}=
    \underbrace{\begin{bmatrix}
      \vekey{H}^0_{\ve{\eta}_1} & & 0   \\
      & \ddots & \\
      0 & & \vekey{H}^0_{\ve{\eta}_N}
      \end{bmatrix}}_{CN \times CN} \,,
  \]}
equation (\ref{eq:log_likelihood_n_taylor_2}) can be rewritten as follows:
\begin{linenomath}{\azul
\begin{align} \label{eq:log_likelihood_n_taylor_3}
l(\vekey{Y} \mid \vekey{\tilde{\eta}}) & \approx    \nonumber\\
& \approx  NC_1 +  \frac{1}{2}[(\vekey{L}^0)^T \vekey{\tilde{\eta}} - (\vekey{L}^0)^T (\vekey{\tilde{\eta}}^0 - (\vekey{H}^0_{\ve{\tilde{\eta}}})^{-1} \vekey{g}^0_{\ve{\tilde{\eta}}})]^T \\
 &\phantom{\approx NC_1 + \frac{1}{2}}
 [(\vekey{L}^0)^T \vekey{\tilde{\eta}} - (\vekey{L}^0)^T (\vekey{\tilde{\eta}}^0 - (\vekey{H}^0_{\ve{\tilde{\eta}}})^{-1} \vekey{g}^0_{\ve{\tilde{\eta}}})] \nonumber\,\,.
\end{align}}
\end{linenomath}
Defining
\begin{linenomath}{\azul
\begin{align}\label{eq:z_var_definition_n}
\ve{\tilde{z}}^0 &:= (\vekey{L}^0)^T (\vekey{\tilde{\eta}}^0 - (\vekey{H}^0_{\ve{\tilde{\eta}}})^{-1} \vekey{g}^0_{\ve{\tilde{\eta}}}) = (\vekey{L}^0)^T \vekey{\tilde{\eta}}^0 - (\vekey{L}^0)^{-1} \vekey{g}^0_{\ve{\tilde{\eta}}} \,,
\end{align}}
\end{linenomath}
we obtain $p(\ve{\tilde{z}}^0 \mid \vekey{\tilde{\eta}})$,
\begin{linenomath}{\azul
\begin{align}
\ve{\tilde{z}}^0 \mid \vekey{\tilde{\eta}} & \sim \mathcal{N}((\vekey{L}^0)^T \vekey{\tilde{\eta}}, \vekey{I}_{CN})\,,
\end{align}}
\end{linenomath}
and the observation matrix $\vekey{Y}$ has been turned into Gaussian conditionally independent pseudo-observations $\ve{\tilde{z}}^0$, a likelihood which \pkg{R-INLA} can deal with.

\section{Calculus for the Dirichlet likelihood}\label{append_likelihood_dirichlet}
In this Appendix, we present all the calculus required for the particular case of the Dirichlet likelihood. For the sake of simplicity, we present those required just for one observation. We start with presenting the likelihood in terms of the linear predictor, we continue with the gradient, followed by the Hessian and finishing with the expected Hessian.
\subsection{Likelihood}
The density function corresponding to Dirichlet distribution has been depicted in expression (\ref{dirichlet}). Using the Dirichlet regression displayed in Equation (\ref{eq:dirichlet_regression}), we know that $\alpha_c = \exp(\eta_c), \ c=1, \ldots, C$. Then the density function or the likelihood for just one observation can be expressed as:
\begin{linenomath}
\begin{equation}
	p(\ve{y} \mid \eta_1, \ldots, \eta_C)= \frac{1}{\text{B}({\exp(\eta_1), \ldots, \exp(\eta_C)})} \prod_{c=1}^C y_c^{\exp(\eta_c) -1} \,.
\end{equation}
\end{linenomath}
Taking logarithms and using the definition of the B function, the next expression is obtained:
\begin{linenomath}
\begin{equation}
\log	p(\ve{y} \mid \vekey{\eta}) = \log \left(\frac{\Gamma(\exp(\eta_1)) \cdots \Gamma(\exp(\eta_C))}{\Gamma({\exp(\eta_1) + \cdots + \exp(\eta_C)})} \right) + \sum_{c=1}^C (\exp(\eta_c) -1) \log(y_c) \,.
    \label{dirichlet_eta}
\end{equation}
\end{linenomath}

\subsection{Gradient}
Here, the gradient of the log likelihood is calculated.
\begin{linenomath}
\begin{align}
  \frac{\partial l}{\partial \eta_c} & =  \exp(\eta_c) \left [\phi(\exp(\eta_c)) - \phi \left(\sum_{c=1}^C \exp(\eta_c)\right) \right] - \exp(\eta_c) \log(y_c) \,,
\end{align}
\end{linenomath}
where $c=1\ldots, C$ and $\phi$ is the digamma function.

\subsection{Hessian}
The second derivatives are calculated for the log likelihood. Let $c$ and $d$ two naturals such as $1 \leq c \leq C$ and $1 \leq d \leq C$, then
\begin{linenomath}
\begin{align}
  \frac{\partial^2 l}{\partial^2 \eta_c} & =  \exp(\eta_c) \left [\phi(\exp(\eta_c)) - \phi \left(\sum_{c=1}^C \exp(\eta_c)\right)\right] + \nonumber \\
  &\phantom{= } + \exp(2 \eta_c) \left [\phi^1(\exp(\eta_c)) - \phi^1 \left(\sum_{c=1}^C \exp(\eta_c)\right)\right] - \nonumber \\
  &\phantom{= } -  \exp(\eta_c) \log(y_c) \,
\intertext{and}
  \frac{\partial^2 l}{\partial \eta_c \partial \eta_d} & = -\exp(\eta_c) \exp(\eta_d) \left[\phi^1 \left(\sum_{c=1}^C \exp(\eta_c)\right)\right] \,, \nonumber
\end{align}
\end{linenomath}
where $\phi$ is the digamma function and $\phi^1$ is the trigamma function.

\subsection{Expected Hessian}
The expected second derivatives are calculated for the log likelihood. Thus
\begin{linenomath}
\begin{align}
  \text{E}\left(\frac{\partial^2 l}{\partial^2 \eta_c} \right) & =  \exp(2 \eta_c) \left [\phi^1(\exp(\eta_c)) - \phi^1 \left(\sum_{c=1}^C \exp(\eta_c)\right)\right]  \nonumber \,\\
\intertext{and}
  \text{E} \left(\frac{\partial^2 l}{\partial \eta_c \partial \eta_d} \right) & = -\exp(\eta_c) \exp(\eta_d) \left[\phi^1 \left(\sum_{c=1}^C \exp(\eta_c)\right)\right] \,, \nonumber
\end{align}
\end{linenomath}
where $\phi$ is the digamma function and $\phi^1$ is the trigamma function.

\end{appendix}

\end{document}